\documentclass[twoside]{article}
\usepackage[a4paper]{geometry}
\usepackage[latin1]{inputenc} %
\usepackage[T1]{fontenc} %
\usepackage{RR}
\usepackage{graphicx}    
\usepackage{listings}   
\usepackage{paralist}        
\usepackage{fmtcount}  
\usepackage{acronym}        
\usepackage{textcomp}
\usepackage{footmisc}
\usepackage{hyperref}
\hypersetup{ 
colorlinks,%
citecolor=black,%
filecolor=black,%
linkcolor=black,%
urlcolor=black
}  
\usepackage{fancyvrb}
\usepackage{macroRap} 
\usepackage{macrogentra2}
\newcommand{\CHRv}{CHR$^{\lor}$}
\newcommand{\CHRt}{OS-CHR$^{\lor}$}

\newcommand{\tuple}[1]{\left\langle #1 \right\rangle}

\newenvironment{prop}{\begin{proposition}\rm}{\end{proposition}}

\RRdate{April 2012}
\RRNo{7939}

\RRauthor{
Armando Gon\c{c}alves da Silva Junior\thanks[sfn]{Federal University of Pernambuco, agsj@cin.ufpe.br}
  \thanks{Work done during internship of Armando Gon\c{c}alves da Silva Junior}%
  \and
Pierre Deransart\thanks{INRIA Paris-Rocquencourt, Pierre.Deransart@inria.fr}%
  \and
Luis Carlos Menezes\thanks{University of Pernambuco, Recife, lcsm@ecomp.poli.br}
  \and
Marcos Aur\'elio Almeida da Silva\thanks{Universit\'e Pierre-Marie Curie, Paris, France, maurelio1234@gmail.com}
  \and
Jacques Robin\thanks{Thal\`es, France, jacques@gmail.com}%
\thanksref{sfn}
}
\authorhead{Gon\c{c}alves \& Deransart \& others}
\RRtitle{Vers une trace g\'en\'erique pour des solveurs de contraintes \`a base de r\`egles}
\RRetitle{Towards a Generic Trace for Rule Based Constraint Reasoning}
\titlehead{Towards a Generic Trace}
\RRresume{
CHR (Constraint Handling Rules) est un langage de programmation adaptable qui permet de sp\'ecifier tr\`es d\'eclarativement des solveurs de contraintes. Un aspect important de leur mise au point concerne leur d\'ebogage. Les implantations actuelles de CHR offrent des possiblilit\'es de traces avec relativement peu d'information.

Dans ce rapport, nous proposons une nouvelle trace CHR qui contient suffisamment d'information pour analyser potentiellement tous les d\'etails d'ex\'ecution de \CHRv, correspondant \`a un niveau d'analyse abstrait et utile, commun \`a diff\'erentes impl\'ementations. 

Cette approche est fond\'ee sur l'id\'ee de trace g\'en\'erique. Une telle trace est d\'efinie comme une extension de la s\'emantique $\omega_r^\lor$ de CHR. On montre qu'elle peut \^etre d\'eriv\'ee de la trace CHR de SWI Prolog. 
}

\RRabstract{
CHR is a very versatile programming language that allows programmers to declaratively specify constraint solvers. An important part of the development of such solvers is in their testing and debugging phases. Current CHR implementations support those phases by offering tracing facilities with limited information. 

In this report, we propose a new trace for CHR which contains enough information to analyze any aspects of \CHRv\  execution at some useful abstract level, common to several implementations. 

This approach is based on the idea of generic trace. Such a trace is formally defined as an extension of the $\omega_r^\lor$ semantics of CHR. We show that it can be derived form the SWI Prolog CHR trace. 
}
\RRmotcle{trace, CHR, CHR$^\vee$, traceur, trace g\'en\'erique, analyseur, s\'emantique observationnelle, d\'eboggage, environnement de programmation, programmation par contraintes, validation}
\RRkeyword{Trace, CHR, CHR$^\vee$, Tracer, Generic Trace, Analysis Tool, Observational Semantics, Debugging, Programming Environment, Constraint Programming, Validation}
\RRprojets{Contraintes}
\URRocq 
\RCParis 

\begin{document}
\makeRR   

\section{Introduction}
\label{introd}

CHR (Constraint Handling Rules)\cite{frue09} is a uniquely versatile and semantically well-founded programming language. It allows programmers to specify constraint solvers in a very declarative way. An important part of the development of such solvers is in their testing and debugging phases. Current CHR implementations support those phases by offering tracing facilities with limited information.

In this report, we propose a new trace for CHR which contains enough information, including source code ones, to analyze any aspects of \CHRv\  execution at some abstract level, general enough to cover several implementations and source level analysis. Although the idea of formal specification based tracer is not new (see for example \cite{jahier00}), the main novelty lies in the generic aspect of the trace. Most of the existing implementations of CHR like in \cite{fru_brisset_highlevel_implementation_techrep95,holzfrueSics98,eclipse-prolog,wielemaker2006swi} include a tracer with specific CHR trace events, but without formal specification, nor consideration with regards to different kind of usages other than debugging.

\vspace{1mm}
The notion of generic trace has been informally introduced and used for defining portable CLP(FD) tracer and portable applications \cite{gentra4cpa,ercimlnai}. We propose here to use this approach to specify a tracer for rule based inference engine like \CHRv. A generic trace has three main characteristics: it is ``high level'' in the sense that it is independent from particular implementations of CHR, it has a specified semantics (Observational Semantics) and can be used to implement debugging tools or applications. 
An important property of the proposed generic trace is that it contains as many information on the solver behaviour as the one contained in the operational semantics. This property is called ``faithfulness'' of the observational semantics.

\vspace{1mm}
In this report, we present a generic trace for \CHRv\  based on its refined operational semantics $\omega_r^\lor$ \cite{kosd06}, 
and describe a first prototype developed for SWI-Prolog \CHRv\  engine. The implementation consists of combining the original trace of the SWI engine with source code information to get generic trace events, and then, allowing the user to filter these events using an SQL-based language.

\vspace{1mm}
This report is organized as follows. Section~\ref{sec:gentra} gives a short introduction to generic traces, observational semantics and faithfulness. Section~\ref{sec:chrd} presents \CHRv, the formal specification of its operational semantics, based on the $\omega_r^\lor$ semantics, and the requirements for the generic trace, as its syntax as well.
Section~\ref{sec:obsem} presents the observational semantics of \CHRv, OS-\CHRv, defining formally the generic trace, and shows its faithfulness. Section \ref{sec:osproto} introduces an executable operational semantics of \CHRv\ defined in SWI Prolog (the code is in the annex) and used to test its formal semantics. Section~\ref{sec:proto} describes the CHR-SWI-Prolog based prototype of the generic trace.
Section~\ref{sec:exp} presents some experimentation. Discussion and conclusions are in the two last sections.

\section{Generic Trace, Observational Semantics and Subtrace}
\label{sec:gentra}

The concept of {\em generic trace} has been first introduced in \cite{ercimlnai}, formally defined in \cite{cp11hal,TMTmanuscript11e}, and a first application to CHR presented in \cite{pierrerafRR09}. A generic trace is a trace with a specification based on a partial operational semantics applicable to a family of processes. We give here its main characteristics and the way to specify a generic trace.

\subsection{Preliminaries}
\label{sec:prelim}

A {\em trace} consists of an initial state $s_0$ followed by an ordered finite or infinite sequence of {\em trace events}, denoted $<s_0, \overline{e}>$. ${\cal T}$ is a set of traces (finite or infinite). A {\em prefix} (finite, of size $t$)  of a trace $T=<s_0, \overline{e_n}>$ (finite or infinite, here of size $n \geq t$) is a partial trace $U_t = <s_0, \overline{e_t}>$ which corresponds to the $t$ first events of $T$, with an initial state at the beginning. 
${\cal T}$ may contain any prefixes of its elements.

A trace can be decomposed into segments containing trace events only, except prefixes which start with a state. An associative operator of concatenation will be used to denote sequences concatenations (denoted $++$). It will be omitted if there is no ambiguity. The neutral element is $[]$ (empty sequence). 
A segment (or prefix) of size 0 is either an empty sequence or a state.


Traces are used to represent the evolution of systems by describing the evolution of their state. 
A state of the system is described by a given finite set of parameters and a state corresponds to a set of values of parameters. Such states will be said {\em virtual} as they correspond to states of the observed system, but they are not actually traced.
We will thus distinguish between actual and virtual traces.

\begin{itemize}
\item the {\em actual traces} (${\cal T}^w$) are a way to observe the evolution of a system by generating traces. The events of an actual trace have the form $e = (a)$ where $a$ is an {\em actual state} described by a set of {\em attributes values}. An actual states is described by a finite set of attributes. Actual traces corresponds to sequences of events produced by a tracer of an observed system. They usually encode virtual states changes in a synthetic manner. 
\item the {\em virtual traces} (${\cal T}^v$) corresponds to the sequence of the virtual states such that for each transition in the system between two virtual states, it corresponds an actual trace event. The virtual trace events have the form $e = (r, s)$ where $r$ is a {\em type of action} associated with a state transition and $s$, called {\em virtual state}, the new state reached by the transition and described by a set of {\em parameters}. Virtual traces correspond to sequences of virtual states of the observed system which produced the actual trace, together with the kind of action which produced the virtual state transition.
\end{itemize}

The correspondence between both kinds of traces is specified by two functions $E: {\cal T}^v \rightarrow {\cal T}^w$ and $I: {\cal T}^w \rightarrow {\cal T}^v$, respectively the {\em extraction} and the {\em reconstruction function}, as illustrated by the figure~\ref{fig:componentsPD3}.

\begin{figure}  
\centering
\includegraphics[width=0.6\linewidth]{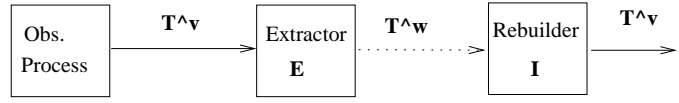}
\caption[Extraction, Reconstruction, Faithfulness Property]
{Extraction, Reconstruction, Faithfulness Property}
\label{fig:componentsPD3}  
\end{figure} 

The idea is that the actual generated trace contains as much information as possible in such way that the virtual trace can be reconstructed from the actual one. In other words, the extraction is done without loss of information. Such a property of the traces is called {\em faithfulness} and, if we denote $Id_v$ (resp. $Id_w$) the identity between virtual traces (resp. actual traces), it states that $E \circ I = Id_v$ (composition) or $E = I^{-1}$, and $I \circ E = Id_w$ (or $I = E^{-1}$).

\vspace{1mm}
Finally, each trace event is numbered by a {\em chrono}, which is an integer incremented by 1 at each new event. It will be ignored in formal presentations, but will be used as a unique identifier of the trace events in the implementations.

\subsection{Components in Trace Design}
\label{sec:quer}

When designing a trace, several components must be taken into consideration. They are depicted in the Figure~\ref{fig:componentsPD4}. 
\begin{figure}  
\centering
\includegraphics[width=0.8\linewidth]{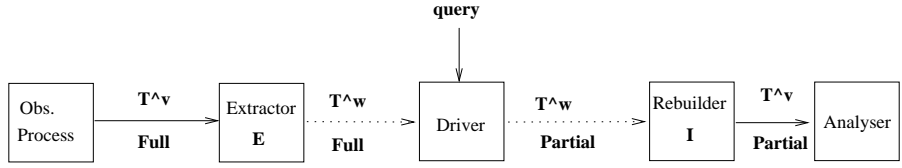}
\caption[Components in Trace Design]
{Components in Trace Design}
\label{fig:componentsPD4}  
\end{figure} 

\begin{enumerate}
\item The observed process whose behavior is modeled by a virtual trace (sequence of successive virtual states) $T^v$.

\item An {\em extractor} component which encodes the virtual trace into the actual one $T^w$. This component corresponds in practice to the tracer formalized by the extraction function $E$.

\item The {\em driver} which realizes the actual trace filtering according to some {\em trace query}. In this report we limit its role to select a subtrace of the so called {\em full trace}.

\item The {\em rebuilder} which may reconstruct from a full or partial actual trace a full or partial virtual trace. This is possible only if there is no loss of information (faithfulness property). The rebuilder is formalized by the reconstruction function $I$.

\item The {\em analyzer}, which corresponds to some debugging tool or particular application, working with the full trace or a partial one.
\end{enumerate}

Notice that in practice the three first components may be interleaved in the sense that for a given query the driver may select directly a subset of the virtual trace, thus avoiding to extract and encode a full actual trace before selecting a subtrace. 

In this report we focus on three components (observed process, extractor and rebuilder) and a property. Their description consists in a faithful observational semantics.

\subsection{Observational Semantics (OS)}
\label{sec:prel}

The evolution of a system defined by its virtual traces and the production of the corresponding actual trace can be described by a so called {\em Observational Semantics} as follows. More general definitions can be found in \cite{TMTmanuscript11e}.

\begin{definition} [Observational Semantics]
\label{def:OS}

An observational semantics consists of $< S, R, A, T, E, I, S_0 >$, where

\vspace{1mm}
\begin{itemize}
\item $S$: {\em domain of virtual states}, a subset of the Cartesian product of domains of parameters.
\item $R$: finite set of {\em action types}, set of identifiers labeling the transitions.
\item $A$: {\em domain of actual states}, a subset of the Cartesian product of domains of attributes.
\item $Tr$: {\em state transition function\footnote{It can be a relation in case of a non deterministic transitions.}} $Tr: R \times S \rightarrow S$, characterized as $Tr(r_i, s_{i-1}) = s_i$, where $(s_{i-1},s_i)$ is the $i_{th}$ transition starting from an initial state $s_0$, and labeled by $r_i$.
\item $E_l$: {\em local trace extraction function} $E_l: S \times R \times S \times {\cal T}^w  \rightarrow A$, which is defined for all $r,s,s'$ such that $Tr(r,s) = s'$ and characterized as \newline$E_l(s_{i-1}, r_i, s_i, T^w_{i-1}) = a_i$, where $a_i$ is the $i^{th}$ actual trace event generated after the initial state $s_0$. In this definition the extraction may use all information accumulated in the already generated actual trace $T^w_{i-1}$.
\item $I_l$: {\em local trace reconstruction function} $I_l: S \times {\cal T}^w  \rightarrow R \times S$ is characterized as $I_l(s_{i-1}, T^w_i) = (r_i, s_i)$. The extraction function is ``local'' if it uses a terminal bounded subsequence of $T^w_i$. Here it will use the last actual trace event only $a_i$\footnote{In general a bounded number of element should be used, but also some forward elements like $a_{i+1}$... For more details see \cite{TMTmanuscript11e}.} ($T^w_i = T^w_{i-1} a_i$).
\item $S_0 \subseteq S$,  set of {\em initial states}.
\end{itemize}
\end{definition}

The local extraction and reconstruction functions can be extended to obtain the functions $E$ (resp. $I$) between sets of virtual and actual traces, as follows: 

\noindent
$E(T^v_n) = s_0 E_l(s_0, r_1, s_1, T^w_0) ... E_l(s_{i-1}, r_i, s_i, T^w_{i-1}) ... E_l(s_{n-1}, r_n, s_n, T^w_{n-1}) = T^w_n$, as $E_l(s_{i-1}, r_i, s_i, T^w_{i-1}) = a_i$ and $T^w_i = s_0 a_1 ... a_i$. 
And 

$I(T^w_n) = s_0 I_l(s_0, T^w_1)... I_l(s_{i-1}, T^w_i)... I_l(s_{n-1}, T^w_n)$ with $I_l(s_{i-1}, T^w_i) = (r_i, s_i)$.


\vspace{1mm}
The Observational Semantics is {\em faithful} if $E$ and $I$ satisfy the faithfulness property, i.e. if $\forall T^v, T^w finite, E(T^v) = T^w  \wedge I(T^w) = T^v$.

\begin{exemple}
\label{ex:miautom}
\ \ 

\vspace{1mm}
The figure~\ref{fig:exmi1} illustrates the following simple automaton (arrows are labeled by type of actions).
\begin{itemize}
\item $S = \{s_0, s_1, s_2\}$, 
\item $R = \{a, b\}$, 
\item $A = \{a, b\}$, 
\item $\forall s, T(b, s) = s_1, T(a, s_1) = s_2$,
\item $E_l(s, r, s') = r$ (the generated actual trace is not used),
\item $\forall s, I_l(s, b) = (b, s_1), I_l(s_1, a) = (a, s_2)$ (the last actual trace event only is used),
\item $S_0 = \{s_0\} $. 
\end{itemize}

The virtual traces are: $s_0 (b,s_1)^+ \{(a,s_2)(b,s_1)^+\}^*$. 
The actual traces are: $s_0 b^+ (ab^+)^*$. 

\begin{figure}[h]
\begin{center}
\includegraphics[width=0.3\linewidth]{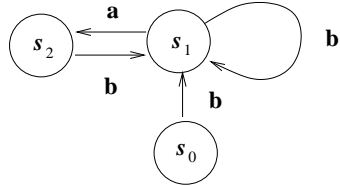}
\end{center}
\caption{Exemple~\ref{ex:miautom}: Finite State Automaton}
\label{fig:exmi1}
\end{figure}

It is straightforward to see that this OS is faithful, using the regular expression form of the traces.
\end{exemple}


\subsection{Faithful Trace Specification}
\label{sec:stetrec}

We first give a condition for an OS to be faithful.

\begin{prop}
\label{prop:faith}
\ \ 

\vspace{1mm}
Given an OS $< S, R, A, Tr, E_l, I_l, S_0 >$, $E$ and $I$ the extensions of $E_l$ and $I_l$ as above, if the following condition holds:

$\forall r, s, s', a, T^w, Tr(r, s) = s' \Rightarrow (E_l(s, r, s', T^w) = a  \wedge I_l(s, T^w a) = (r, s'))$, where $T^w$ is a finite actual trace generated by successive applications of the transition function $Tr$ from an initial state $s_0$ until an ultimate state $s$,

then the OS is faithful, i.e. $\forall T^v, T^w finite,  E(T^v) = T^w  \wedge I(T^w) = T^v$, where $T^v$ is a virtual trace corresponding to successive applications of the transition function $Tr$ from an initial state $s_0$ until an ultimate state $s$.
\end{prop}

\begin{proof}
\label{proof:faith}
\ \ 
The proof is by induction on the size of the traces. By definition with each new transition actual and virtual traces are increased by one event only.
Initially: (traces of size 0) $E(s_0) = s_0  \wedge I(s_0) = s_0$.

\vspace{1mm}
Let us assume the property holds for a trace of size $n$: 

 (1) $E(T^v_n) = T^w_n \wedge I(T^w_n) = T^v_n$ (2).

(1) means, according to the definition of $E$, that $T^w_n$ is the actual trace generated by $n$ successive applications of $Tr$ function from a state $s_0$ until some ultimate state $s_n$ by using the local extraction function $E_l$ at each step.

(2) means, according to the definition of $I$, that $T^v_n$ is the virtual trace reconstructed from the actual trace $T^w_n$. The holding conjunction means that the reconstructed virtual trace is the same as the one used to produce the actual one, therefore $I(T^w_n)$ ends in a state $s_n$, the same as for $T^v_n$.

One shows that (${++}$ concatenation and list constructor are omitted for singletons)

(1') $E(T^v_{n+1}) = T^w_{n+1} \wedge I(T^w_{n+1}) = T^v_{n+1}$ (2')  holds. 

Assume that the ultimate state reached by $T^v_n$ is $s_n$. Then by definition (if $s_n$ is not a final state), $\exists r_{n+1} s_{n+1}, Tr(r_{n+1}, s_n, s_{n+1})$ holds, hence:
$T^v_{n+1} = T^v_n (r_{n+1}, s_{n+1})$, and, by definition $T^w_{n+1} = T^w_n a_{n+1}$.

\vspace{1mm}
(1') holds: 

\noindent
$E(T^v_{n+1}) = E(T^v_n (r_{n+1}, s_{n+1})) = T^w_n E_l(s_n, r_{n+1}, s_{n+1}, T^w_n) = T^w_n a_{n+1} = T^w_{n+1}$, by hypothesis (1) and where the last argument of $E_l$, $T^w_n$, is the actual trace resulting from $E(T^v_n)$.

(2') holds: 

$I(T^w_{n+1}) = I(T^w_n a_{n+1}) =  T^v_n I_l(s_n,T^w_n) = T^v_n (r_{n+1}, s_{n+1}) = T^v_{n+1}$, by hypothesis (2) and the definition of $I_l$.

\end{proof}

The proposition~\ref{prop:faith} shows how to design a faithful observational semantics: by specifying together the local extraction and reconstruction functions, verifying at each step defined by the transition function that the generated trace allows to reconstruct the same reached state as the one specified by the transition function.

This results can be extended to actual subtraces (defined as traces obtained by considering a subset of attributes of the actual traces), provided that there are sufficiently many attributes in the subtrace to recompute parameters (then the corresponding virtual subtrace is a virtual trace with a subset of parameters)\footnote{This implies some restrictions on the kind of queries (in case of dependencies between parameters or attributes) which are not detailed here. For more details see \cite{TMTmanuscript11e}.}.

This is illustrated in the figure~\ref{fig:componentsPD4}. A query applied to the actual trace selects a partial actual trace in such a way that the resulting partial trace can be reconstructed as a partial virtual trace (the one from which the partial actual trace could be extracted).
\vspace{1mm}
In practice, the generic trace specification consists of an operational semantics corresponding to some abstract level of process observation, instrumented to produce an actual trace. The level of description (granularity of the events) should be chosen in such a way that this abstract operational semantics can be abstracted from each particular semantics of each process of the family. Symmetrically, it is requested that the abstract operational semantics can be ``implemented'' in each process of the family. 

The faithfulness property of the observational semantics guarantees that the generic actual trace preserves the whole information concerning the process behavior, in such a way that all desirable properties of the semantics can be studied just by looking at the trace.

\section{Generic Trace for \CHRv}
\label{sec:chrd}

In this section we introduce the generic trace proposed for \CHRv.
It is based on the refined Theoretical Operational Semantics for CHR, $\omega_r^\lor$, as defined in \cite{kosd06}.

Such semantic is declarative enough to cover most of the CHR implementations. It is the case for ECLiPSe Prolog \cite{eclipse-prolog} and SWI-Prolog \cite{wielemaker2006swi} whose operational semantics can be viewed as a refinement of $\omega_r^\lor$ (conversely $\omega_r^\lor$ can be viewed as an abstraction of the semantics of these implementations).

\subsection{Introducing \CHRv}
\label{ssec:illustr}

Constraint Handling Rules emerges in the context of Constraint
Logic Programming (CLP) as a language for describing Constraint Solvers. In
CLP, a problem is stated as a set of constraints, a set of predicates and a
set of logical rules. Problems in CLP are generally solved by the interaction
of a logical inference engine and constraint solving components. The logical
rules (written in a host language) are interpreted by the logical inference
engine and the constraint solving tasks are delegated to the constraint
solvers. We borrow from \cite{pierrerafRR09} some elements of presentation.

The following rule base handles the less-than-or-equal problem: 
\lstset{
	backgroundcolor==\color{darkgray},
	tabsize=4,
	rulecolor=,
	language=Prolog,
        basicstyle=\footnotesize,  
        upquote=true,
        aboveskip={1.5\baselineskip},
        columns=fixed,
        showstringspaces=false,
        extendedchars=true,
        breaklines=true,
        prebreak = \raisebox{0ex}[0ex][0ex]{\ensuremath{\hookleftarrow}},
        showtabs=false,
	numberstyle=\tiny,
	numbersep=5pt,
    showspaces=false,	
    showstringspaces=false,
    identifierstyle=\ttfamily,
    keywordstyle=\color{black}\bfseries\emph,
    commentstyle=\color[rgb]{0.133,0.545,0.133},
    stringstyle=\color[rgb]{0.627,0.126,0.941},
}

\begin{lstlisting}
antisymmetry @ leq(X,Y), leq(Y,X) <=> X = Y. 
reflexivity @ leq(X,X) <=> true. 
idempotence @ leq(X,Y) \ leq(X,Y) <=> true. 
transitivity @ leq(X,Y), leq(Y,Z)  ==> leq(X,Z).
\end{lstlisting}

\vspace{1mm}
This CHR program specifies how $leq$ simplifies and propagates as a constraint.
The rules implement reflexivity, antisymmetry, idempotence and transitivity in a
straightforward way. The $reflexivity$ rule states that $leq(X,Y)$ simplifies to
$true$, provided it is the case that $X=Y$. This test forms the (optional)
guard of a rule, a precondition on the applicability of the rule. Hence,
whenever we see a constraint of the form $leq(X,X)$ we can simplify it to true.

\vspace{2mm}
The $antisymmetry$ rule means that if we find $leq(X,Y)$ as well as $leq(Y,X)$
in the constraint store, we can replace it by the logically equivalent $X=Y$.
Note the different use of $X=Y$ in the two rules: in the $reflexivity$ rule the
equality is a precondition (test) on the rule, while in the $antisymmetry$ rule
it is enforced when the rule fires. (The reflexivity rule could also have been
written as $reflexivity @ leq(X,X) <=> true.$) 

\vspace{2mm}
The rules $reflexivity$ and $antisymmetry$ are \textit{simplification rules}. In
such rules, the constraint found are removed when the rule applies and fires.
The rule $idempotence$ is a \textit{simpagation rule}, only the constraint in the right part
of the head will be removed. The rule says that if we find $leq(X,Y)$ and another
$leq(X,Y)$ in the constraint store, we can remove one.

\vspace{2mm}
Finally, the $transitivity$ rule states that the conjunction $leq(X,Y), leq(Y,Z)$
implies $leq(X,Z)$. Operationally, we add $leq(X,Z)$ as (redundant) constraint.
without removing the constraints $leq(X,Y), leq(Y,Z)$. This kind of rule is
called \textit{propagation rule}.

\vspace{3mm}
The CHR rules are interpreted by a CHR inference engine by rewriting the initial
set of constraints by the iterative application of the rules. Its extension with disjunctive
bodies, CHR$^\vee$ boosts its expressiveness power, turning it into a general
programming language (with no need of an host language).

\vspace{2mm}
Summarizing, there are three kinds of rules in CHR as in \CHRv: simplification, propagation and simpagation. 

\vspace{1mm}
The simpagation rules are the most general category of rules and they have the following form $r @ H_k \backslash H_r \Leftrightarrow G | B.$, where \verb|r| is an identifier for the rule, $H_r$ and $H_k$ are the heads of the rule (in the following they will be denoted respectively as {\em keep} and {\em remove}), $G$ is the guard, and $B$ is the body. If the guard is \verb|true|, it can be omitted. 

\vspace{2mm}
The operational semantics of such rule is that if $H_k$ and $H_r$ are found in the constraint store and the guard $G$ is entailed by it, the constraints in $H_r$ should be removed and the constraints in the body $B$ should be added to the constraint store. If $H_k$ is empty, this rule is called a \emph{Simplification Rule}, and this part of the rule is omitted. On the other side, if $H_r$ is empty, this rule is called a \emph{Propagation Rule}. In this case, the second part of the head of the rule is omitted and the $\Leftrightarrow$ is replaced by the symbol $\Rightarrow$.

\vspace{1mm}
In \CHRv, we distinguish two kinds of constraints: Rule defined constraints (RDC), which are declared in the current program and defined by CHR rules, and built-in constraints (BIC), which are predefined in the host language or imported CHR constraints from some other module. Furthermore, in \CHRv, bodies may contain disjunctions.

\begin{exemple}
\label{ex-append}

The following \CHRv rule defines the \verb|append(X,Y,Z)| constraint:

\begin{verbatim}
r1 @ append(X,Y,Z) <=> (X = [], Z = Y);
                       (X = [H|L1], Z = [H|L2], append(L1,Y,L2)).
\end{verbatim}

In this rule, \verb|Z| is a list composed by the elements of the list \verb|X| followed by the elements of the list \verb|Y|. If \verb|append(X,Y,Z)| holds, we have two options: (i) \verb|X=[]| and, therefore, \verb|Z=Y|; or (ii) \verb|X| is a list in the form \verb+[H|L1]+, and thus, \verb|Z| is composed by \verb|H| followed by \verb|L1| and then followed by the elements in \verb|Y|.

We illustrate the \CHRv solver behavior, by showing the effect of successive applications of the \CHRv rules on the constraint store.
 Let us suppose the initial state of the constraint store is $append([1], [2], Z)$. The actual execution of this
rule base is represented by the following set of transitions. The notation $\tuple{G_0};\ldots;\tuple{G_n}$ represents the alternative stores. The transition $\mapsto_{r_1}$ represents the application of $r_1$ and the transition $\mapsto_{*}$ the removal of failed stores. Notice that from the final state we can conclude $Z = [1,2]$:

{\small
  \begin{eqnarray*}
	\tuple{append([1], [2], Z)} & \mapsto_{r_1} & \\
	\tuple{[1] = [], Z = [2]} ; \tuple{[1] = [1|[]], Z = [1|L2], append([], [2], L2)} & \mapsto_{*} &\\
	\tuple{[1] = [1|[]], Z = [1|L2], append([], [2], L2)} & \mapsto_{r_1} &\\
	\tuple{[1] = [1|[]], Z = [1|L2], [] = [], L2 = [2]} ;\\ \tuple{[] = [H|L1], L2 = [H|L2'], append(L1, [2], L2') )} & \mapsto_{*} &\\
	\tuple{[1] = [1|[]], Z = [1|L2], [] = [], L2 = [2]} &  &
  \end{eqnarray*}
}

\end{exemple}

\subsection{Operational Semantics $\omega_r^\lor$}
\label{ssec:opsem}

The $\omega_r^\lor$ semantics given here is adapted from \cite{frue09}.

\vspace{1mm}
We define $CT$ as the constraint theory which defines the semantic of the built-in constraints
and thus models the internal solver which is in charge of handling them. We assume it supports
at least the equality built-in. We use $[H|T]$ to indicate the first ($H$) and the remaining ($T$)
terms in a list or stack, $+$ for pushing elements into stack (there may be several pushed elements represented as a list), $++$ for sequence concatenation and $[]$ for empty sequences. 
We use the notation $\{a_0,\ldots,a_n\}$ for both bags and sets. Bags are sets which allow repeats.
We use $\cup$ for set union and $\uplus$ for bag union, $\{\}$ to represent both the empty
bag and the empty set, and $E-E'$ to remove all occurrences of elements of $E'$ from $E$.
The identified constraints have the form $c\#i$, where $c$ is a user-defined constraint and $i$ a natural number. They differentiate among copies of the same constraint in a bag. We also use the functions $chr(c\#i)=c$ and $id(c\#i)=i$ with their natural extension to lists of constraints and of identifiers.

\vspace{1mm}
A CHR program is a sequence of rules, and the head and body of a rule are considered sequences of atomic constraints.
A number is associated with every atomic head constraint, called the \textit{occurrence}. Head constraints are numbered per functor, starting from 1,in top-down order from the first rule to the last rule, and from left to right.
However, removed head constraints in a simpagation rule are numbered before kept head constraints. This numbers will be used to show in witch position the active constraint is (the $j$ indicator).

\vspace{1mm}
An execution state ${\cal E}$ is a tuple $\tuple{A,S,B,T}_n$, where 
\begin{itemize}
\item $A$ is the execution stack; 
\item $S$ is the UDCS (User Defined Constraint Store), a bag of identified user defined constraints; 
\item $B$ is the BICS (Built-in Constraint Store), a conjunction of constraints; 
\item $T$ is the Propagation History, a set of sequences for each recording the identities of the user-defined constraints which fired a rule; 
\item $n$ is the next free natural used to number an identified constraint. 
\end{itemize}

Current alternatives are denoted as ordered sequence of execution states, ${\cal L}=[{\cal E}_1,{\cal E}_2,...{\cal E}_n]$  where ${\cal E}_1$ is the active execution state and $[{\cal E}_2,...,{\cal E}_n]$ the remaining alternatives. 

The initial configuration is represented by ${\cal L}_0 = [\tuple{A, \{\}, true, \{\}}_{1}]$. The top of execution stack $A$  is a ``goal'' constraint (the one corresponding to the initial goal of the program) that will be processed.
The transitions are applied non-deterministically until no transition is applicable anymore. 

\vspace{1mm}
The formal description of the transitions is done in the form of rules: for each type of action $r \in R$ there is a rule of the form {\bf r} $s \mapsto s'$, such that $Tr(r,s) = s'$. 

\begin{center}
{\small
\fbox{
\begin{minipage}[c]{12cm}

\begin{description}
 \item [Solve+Wake] $[\tuple{[c|A], S, B, T}_n | {\cal L}] \mapsto [\tuple{A'+A,S, c \land B, T}_n | {\cal L}]$, where $c$ is built-in and $A' = wakeup(S,c,B)$ where $wakeup$ is a function that implements the \textit{wake-up policy}\cite{frue09} whose result is a list of constraints of $S$ woken by adding $c$ to $B$.

 \item [Activate] $[\tuple{[c|A], S, B, T}_n | {\cal L}] \mapsto [\tuple{[c\#n:1|A], c \uplus S,  B, T}_{n+1} | {\cal L}]$, where $c$ is user-defined constraint.

 \item [Reactivate] $[\tuple{[c\#i|A], S, B, T}_n | {\cal L}] \mapsto [\tuple{[c\#n:1|A], S,  B, T}_{n}| {\cal L}]$, where $c$ is user-defined constraint.

 \item [Apply.1] $[\tuple{[c\#i:j|A], H_1 \uplus H_2 \uplus S, B, T}_n | {\cal L}] \mapsto [\tuple{[c\#i:j|A], H_1 \uplus H_2 \uplus S, B, T}_n | {\cal L}]$ where the $j^{th}$ occurrence of a constraint with same functor as $c$  exists in the head of a fresh variant of a rule $r @ H'_1 \backslash H'_2 \Leftrightarrow g | C$ and a matching substitution modulo $B,e$\footnote{The matching substitution will record all Equalities in the Built-in Store\label{fn:repeat}}, such that $chr(H_1) = e(H_1')$, $chr(H_2) = e(H_2')$ and  $\{(r,id(H_1){++}id(H_2))\} \notin T$. 

 \item [Apply.2] $[\tuple{[c\#i:j|A], H_1 \uplus H_2 \uplus S, B, T}_n | {\cal L}] \mapsto [\tuple{C+H+A, H_1 \uplus S, e \land B, T'}_n | {\cal L}]$  where the $j^{th}$ occurrence of a constraint with same functor as $c$  exists in the head of a fresh variant of a rule $r @ H'_1 \backslash H'_2 \Leftrightarrow g | C$ and a matching substitution modulo$B,e$\footref{fn:repeat}, such that $chr(H_1) = e(H_1')$, $chr(H_2) = e(H_2')$ and  $\{(r,id(H_1){++}id(H_2))\} \notin T$ and $CT \models \exists (B) \land \forall B \supset \exists (e \land g)$ and $T' = T \cup \{(r,id(H_1){++}id(H_2))\}$, and $H = c\#i:j$ if $c$ is in $H_1$, $H = []$ if $c$ is in $H_2$,  and c $\in H_1 \underline{\vee} H_2$.

 \item [Drop] $[\tuple{[c\#i:j|A], S, B, T}_n | {\cal L}] \mapsto [\tuple{A, S, B, T}_{n}| {\cal L}]$, where there is no occurrence $j$ for $c$ in the program.

 \item [Default] $[\tuple{[c\#i:j|A], S, B, T}_n | {\cal L}] \mapsto [\tuple{[c\#i:j+1|A], S, B, T}_{n} | {\cal L}]$, if no other transition is possible in the current state.

 \item [Split] $[\tuple{[c_1\lor ... \lor  c_m|A], S, B, T}_n |{\cal L}] \mapsto [\sigma_1,..., \sigma_m|{\cal L}]$, where $\sigma_i = \tuple{[c_i|A], S, B, T}_{n} $, for $1 \leq i \leq m$. This transition implements depth-first, other search strategies can be implemented by easily changing this definition.

 \item [Fail] $[{\cal E}|{\cal L}] \mapsto {\cal L}$, This transition is called automatically if ${\cal E}$ is a failed state. By definition a failed state occurs when the Built-in store is false.

\end{description}
\end{minipage}
}}
\end{center}

This semantics is adapted from \cite{frue09}, it differs mainly by two additional type of actions: {\bf Apply.1} ({\bf Apply.2} corresponds to the original {\bf Apply} rule) and {\bf Fail}. {\bf Apply.1} corresponds to an attempt of using a CHR rule and will be applied only once for each $j, H1,H2$ occurrence. If the transition {\bf Apply.2} is applied, there was an application of {\bf Apply.1} for the same occurrence $j$, but the converse does not hold. The {\bf Fail} case corresponds to a failed computation. In this case, other alternatives are explored.

\subsection{Generic Trace}
\label{gt-omega-t}

We introduce here the generic actual trace of \CHRv\ informally. Each transition in the $\omega_r^\lor$ semantics should generate an actual trace event. Here are the expected types of actions together with the other possible attributes. There are two categories of attributes: those coding the virtual states which may serve for the reconstruction, some others may bring useful information for potential applications.

Each transition corresponding to the type of action $r$ generates an actual trace event whose first attribute characterizes $r$ uniquely. A subset of the attributes only is attached to a specific type of event. Attributes are as follows:

\begin{itemize}
\item {\bf Port}, or {\bf Event Name: p}. It belongs to 

\noindent
\{$Wake, ActivateRDC, ReactivateRDC, TryRule, ApplyRule, Drop,$ \\ $Default, Split, Fail$\}.

 There is an obvious bijection between $R$ (the set of type of actions) of the OS and this set of event names. It belongs to all events and is always the first attribute.
\item {\bf Constraint instance of some type w} where the constraint is a term $c$ before the first activation, or $c\#i:j$ in the other cases, where $i$ and $j$ are integers. A constraint will be represented by a list of form $[p, t_1, t_2, ..., t_m]$ or $[p, t_1, t_2, ..., t_m, i, j]$. $w$ is $cons$ in the former case and $cinst$ in the later.
\item {\bf List of constraints instances of some type w: w-lc}.  It may be a list of terms or a list of the form $[w,[[c_1], ...,[c_m]]]$, each element represented by a list as above. {\bf w} is the kind of constraints. It may be $woken$, $addrdc$, $addbic$, $keep$, $remove$ or $guard$ constraints, or a matching substitution $match$, represented by a list of equations.
\item {\bf Reference to a previous trace event: ref}, where {\bf ref} has the form @$i$ and $i$ is an integer identifying a previous actual trace event (a previous chrono). 
\item {\bf Rule name: r@}, where {\bf r@} is the name of a rule in the source program.
\item {\bf State number: n}. An integer corresponding to the parameter numbering a solver virtual state. It belongs to all events and is always the last attribute.
\end{itemize}

This is formalized in the following tables. The first table is a context free syntax of the generic trace. Terminals are between brackets. Non terminals corresponding to attributes are names of attributes. Optional items are between square brackets (those which are not in terminals), options are separated by vertical bars. Brackets after $\in$ denote a set (to avoid multiplicity of rules).

\begin{table} [h] 
\begin{tabular}[t]{|l|l|}
\hline Gen. traceT & ::= Ev$\_1$...Ev$\_m$, $m \geq 0$\\
\hline Ev & ::=  \{{\tt GT: [}\} Chrono \{{\tt ,}\} Port [Other] Spec$\_$Attr \{{\tt ,}\}  \\
\hline  &   StateNumber \{{\tt ]}\}  \\
\hline Chrono  & ::= \{integer\} \\
\hline Port & $\in$ \{ $Wake, ActivateRDC, RectivateRDC, TryRule,$  \\
\hline      &     $ApplyRule, Drop, Default, Split, Fail$ \}\\
\hline Other  & ::= Ref | RuleName | Indice | CT | CTIJ \\
\hline Ref  & ::= \{, @ integer\}  \\ 
\hline RuleName  & ::= \{, identifier @\}  \\ 
\hline Indice  & ::= \{, integer\} \\
\hline Spec$\_$Attr & ::= $\epsilon$ | \{,\} SpecAttr  Spec$\_$Attr  \\
\hline SpecAttr & ::= \{[\} Attr ListCT \{]\}  \\
\hline Attr & $\in$ \{ {\tt Woken}, {\tt Addrdc}, {\tt Addbic}, {\tt Keep}, {\tt Remove},  \\
\hline   &  {\tt Guard}, {\tt Match} \}  \\
\hline ListCT & ::= $\epsilon$ | \{,\} (CT | CTIJ)  ListCT)  \\
\hline CT  & ::= \{[ predicate, t\_1, ..., t\_n ]\}  \\ 
\hline CTIJ & ::= \{[ predicate, t\_1, ..., t\_n, i, j ]\}  \\
\hline StateNumber  & ::= \{integer\} \\
\hline
\end{tabular}
\caption[Generic Trace Grammar]{Syntax of the Generic Trace Grammar}
\label{table:GenTrGrammar}  
\end{table}

The following table gives the list of the attributes other than {\tt Chrono}, {\tt Port} and {\tt StateNumber} for each type of event.
The firts item is the transition name ($\in R$), the second item the corresponding port (value of the attribute {\tt Port}) and the last item is the list of attributes.

\begin{table} [h] 
\begin{tabular}[t]{|l|l|l|}
\hline $r \in R$ & Port &  Specific Attributes \\
\hline 
\hline Solve+Wake & $Wake$ &  {\tt Cons}, {\tt Woken} \\
\hline Activate & $ActivateRDC$ & {\tt Cinst}  \\
\hline Reactivate & $ReactivateRDC$ & {\tt Cinst}, {\tt Ref}  \\
\hline Apply.1 & $TryRule$ &  {\tt RuleName}, {\tt Cinst}, {\tt Keep}, {\tt Remove}, {\tt Guard} \\
\hline Apply.2 & $ApplyRule$ & {\tt Ref}, {\tt Addrdc}, {\tt Addbic}, {\tt Remove}, {\tt Match}, {\tt Cinst}  \\
\hline Drop & $Drop$ &  {\tt Cinst} \\
\hline Default & $Default$ & {\tt Index}\\
\hline Split & $Split$ & {\tt Ref} \\
\hline Fail & $Fail$ &  {\tt Ref}\\
\hline
\hline
\end{tabular}
\caption[Generic Trace Grammar]{Specific Attributes for each type of event (except chrono and state number)}
\label{table:Attributes}  
\end{table}

The following lists all attributes (specific or other) for each actual trace event corresponding to some type of action. All examples\footnote{The Prototype does not generate the $j$ index of the constraint, therefore there isn't any example with this info.} are extracted from a generic trace of the section~\ref{ssec:oschrproto}. 

\begin{itemize}
\item {\bf Solve+Wake} (port $Wake$):  {\em a built-in constraint (BIC) is solicited and some constraints are woken}\\
2 attributes (total 5):
\begin{itemize}
\item {\tt Cons}: the built-in constraints being ``executed''.
\item {\tt Woken}: the (possibly empty) list of the constraints woken by the \textit{wake-up policy}.
\end{itemize}

Example: {\tt GT: [61,Wake,[=,C1,a0],[woken,[[node,r1,C1,359]]],360]} 

\item {\bf Activate} (port $ActivateRDC$):  {\em activate a Rule Defined Constraint (RDC) getting the RDC $c\#i:j$ from the top of the execution stack and activating it.}\\
1 attribute (total 4):
\begin{itemize}
\item {\tt Cinst}: the user defined constraint which is ``introduced'' and ``executed''. The attribute value is the created instance of this constraint.
\end{itemize}

Example: {\tt GT: [3,ActivateRDC,[edge,r1,r9,332],333]}

\item {\bf Reactivate} (port $ReactivateRDC$): {\em Activate a Rule defined constraint with justification}\\
2 attributes (total 5):
\begin{itemize}
\item {\tt Cinst}: the user-defined constraints instance being ``re-executed'', which became active.
\item {\tt Ref}: a reference to the {\bf Solve+Wake} event where this constraint has been woken (form of justification).
\end{itemize}

Example: {\tt GT:[62,Reactivate,[node,r1,a0,359],@61,360]}

\item {\bf Apply.1} (port $TryRule$): {\em attempt to apply a Rule; it may be followed by an {\bf Apply.2} event in case of successful application, or another event in case it cannot by applied.}\\
5 attributes (total 8):
\begin{itemize}
\item {\tt RuleName}: the name of the tried rule in the source code.
\item {\tt Cinst}: the active user-defined constraint instance (the one at the top of the execution stack).
\item {\tt Keep}: the {\em keep} constraints of the store used to match the head of the tried rule.
\item {\tt Remove}: the {\em remove} constraints of the store used to match the head of the tried rule.
\item {\tt Guard}: the {\em guard} constraints of the tried rule (as in the source program). This information may be useful in case of failure.
\end{itemize}

Example: {\tt GT: [102,$TryRule$,wrong@,[node,r5,a0,375],
[keep,[[node,r4,a0,371], [edge,r4,r5,340], [node,r5,a0,375]]],
[remove,[]],[guard,[[=,a0,a0]]],376]}

\item {\bf Apply.2} (port $ApplyRule$):  {\em applying the rule with success (true guard)}\\
7 attributes (total 10):
\begin{itemize}
\item {\tt Ref}: a refence to the previous {\bf Apply.1} event where the name of the applied rule, the active constraint and some other information can be found.
\item {\tt Addrdc}: the user-defined constraints instances of the body of the applied rule, pushed on the stack.
\item {\tt Addbic}: the built-in constraints instances of the body of the applied rule, pushed on the stack.
\item {\tt Keep}: the {\em keep} constraints of the store used to match the head of the tried rule.
\item {\tt Remove}: the {\em remove} constraints of the store used to match the head of the tried rule.
\item {\tt Match}: the successful matching equations (a way to give the current substitution).
\item {\tt Cinst}: the active user-defined constraint instance (the one at the top of the execution stack).
\end{itemize}

Example: {\tt GT: [103,$ApplyRule$,@102,[addrdc],[addbic,[fail,fail]],
[keep,[[node,r4,a0,371], [edge,r4,r5,340], [node,r5,a0,375]]],
[remove,[]],[match,[edge(Ri,Rj)=node(,r4,a0)],
[node(Ri,Ci)=edge(,r4,r5)], [node(Rj,Cj)=node(,r5,a0)]],
[node,r5,a0,375],376]}

\item {\bf Drop} (port $Drop$):  {\em drop a constraint. The currently active constraint $c\#i:j$ is removed from the stack. There is no more occurrences $j$ for $c$ in the program.}\\
1 attribute (total 4):
\begin{itemize}
\item {\tt Cinst}: the constraint which is popped from the execution stack.
\end{itemize}

Example: {\tt GT: [4,$Drop$,[edge,r1,r9,332],333]}

\item {\bf Default} (port $Default$): {\em  The occurrence index $j$ of the active constraint $c\#i:j$ is incremented (proceed to the next occurrence of the constraint instances in the program). There is no more occurrences $j$ for $c$ in the program.}\\
2 attributes (total 5):
\begin{itemize}
\item {\tt Cinst}: the last used constraint occurrence.
\item {\tt Index}: The occurrence Index $j$ is incremented (proceed to the next occurrence of the constraint instances in the program.
\end{itemize}

Example: {\tt GT: [28, $Default$, [node, r7, r, 5, 386], 6, 388]}\footnote{It's a hand-made example. Although, It's not difficult to compute how many Default transitions was applied before the Drop or Apply.2. For the prototype this transition is negligible.}

\item {\bf Split} (port $Split$): {\em create a disjunction. Occurs when a rule is disjunctive.} \\
1 attribute (total 4):
\begin{itemize}
\item {\tt Ref}: a reference to the most recent {\bf Apply.2} event with the rule whose body contains the disjunction.
\end{itemize}

Example: {\tt GT: [60,$Split$,@59,360]}

\item {\bf Fail} (port $Fail$): {\em the referred rule application fails. It Occurs when the Built-In store has been tested false.}\\
1 attribute (total 4):
\begin{itemize}
\item {\tt Ref}: a reference to the most recent failed {\bf Apply.2} event.
\end{itemize}

Example: {\tt GT: [104,$Fail$,@103,376]}

\end{itemize}

All the variables which occur in the initial goal will keep their original name in all their occurrences in the generic trace. 
Each actual trace event has a unique identifier, called the {\bf chrono}, which is an integer incremented by 1 at each new event.

\vspace{1mm}
The formal definition of trace generation will be given in section~\ref{fc-omega-t}. 

\section{Observational Semantics of \CHRv (OS-\CHRv)}
\label{sec:obsem}

We specify the observational semantics of \CHRv, OS-\CHRv, on the top of the operational semantics of section~\ref{ssec:opsem}, by specifying the transition function and the local functions of extraction $E_l$ and of reconstruction $I_l$. The resulting OS, called OS-\CHRv, is faithful by construction (by property~\ref{prop:faith}).


According to the definition~\ref{def:OS}, the observational semantics consists of $< S, R, A, Tr, E_l, I_l, S_0 >$. The definition of the transition function $Tr$ is given by the operational semantics. The definitions of $E_l, I_l$ are given in the next two sections. The others elements are as follows.

\vspace{1mm}
\begin{itemize}
\item $S$: {\em domain of virtual states}. It the set of configurations defined as a list of execution states, where an execution state ${\cal E}$ is defined by the $\tuple{A,S,B,T}_n$ as described in the section~\ref{ssec:opsem}.
\item $R$: finite set of {\em action types}: \{{\bf Solve+Wake, Activate, Reactivate, Apply.1, Apply.2, Drop, Default, Split, Fail}\}
\item $A$: {\em domain of actual states}: each state consists of a tuple of values of a subset of the attributes. All attributes are defined in the section~\ref{gt-omega-t}.
\item $S_0 \subseteq S$,  set of {\em initial states}, specified below.
\end{itemize}
In this presentations the chrono is omitted.

\subsection{OS-\CHRv: Extraction ($Tr, E_l$)}
\label{fc-omega-t}

This description is based on the transitions as described in the table~~\ref{ssec:opsem}. Each item corresponds to a type of action and it specifies the new generated actual trace event, using the previous state, the reaches state and the previously generated actual trace. The current generated actual trace is denoted $N$, an ordered sequence of the trace events. 
It has the form, for each type of action $r$ and transition $Tr(r, s) = s'$: $E_l(s, r, s', N) = a$, which will be represented by a rule:
\begin{center}
{\bf r} $s, N \mapsto s', N {++} a$
\end{center}

\vspace{1mm}
The initial configuration will be represented as 

${\cal L}_0 = [\tuple{A, \{\}, true, \{\}}_{1}] ,N=[]$

\begin{center}
{\small
\fbox{
\begin{minipage}[c]{12cm}
\begin{description}
 \item [Solve+Wake]$\\$
$[\tuple{[c|A], S, B, T}_{n} | {\cal L}], N \mapsto [\tuple{A'+A,S, c \land B, T}_{n} | {\cal L}],$
$N{++}[Wake, c, wakeup(S,c,B), n]$, where $SolveCond$. 

 \item [Activate] $[\tuple{[c|A], S, B, T}_{n} | {\cal L}],N \mapsto [\tuple{[c\#n:1|A], c \uplus S,  B,T}_{n+1} | {\cal L}],$\\ $N{++}[ActivateRDC, c, n]$, where $c$ is a rule-defined constraint

 \item [Reactivate] $[\tuple{[c\#i|A], S, B, T}_n | {\cal L}], N \mapsto [\tuple{[c\#n:1|A], \{c\#n\} \uplus S,  B, T}_{n} | {\cal L}]$, \\
$N{++}[ReactivateRDC,c,wake(c,N)]$, where $CondReac$ (see below)

 \item [Apply.1] $[\tuple{[c\#i:j|A], H_1 \uplus H_2 \uplus S, B, T}_n | {\cal L}], N \mapsto [\tuple{[c\#i:j|A], H_1 \uplus H_2 \uplus S, B, T}_n | {\cal L}]$, 
$N{++}[TryRule, r, {c\#i:j}, H_1, H_2, g, n] $ where $CondApp1$ (see below).

 \item [Apply.2]  $[\tuple{[c\#i:j|A], H_1 \uplus H_2 \uplus S, B, T}_n| {\cal L}], N \mapsto [\tuple{C+[H|A], H_1 \uplus S, e \wedge B, T'}_n| {\cal L}], $\\
$N{++}[ApplyRule, tryRule(N), addRDCs(C), addBICs(C), H_1, H_2 ,e , H, n]$
where $CondApp2$ (see below).

 \item [Drop] $[\tuple{[c\#i:j|A], S, B, T}_n | {\cal L}], N \mapsto [\tuple{A, S, B, T}_{n} | {\cal L}], N{++}[Drop, {c\#i:j}, n]$. Where c is an active constraint.

 \item [Default] $[\tuple{[c\#i:j|A], S, B, T}_n | {\cal L}], N \mapsto [\tuple{[c\#i:j+1|A], S, B, T}_{n} | {\cal L}], N{++}$\\
 $[Default, c\#i:j, {j+1}, n].$

 \item [Split] $[\tuple{[c_1\lor ... \lor  c_m|A], S, B, T}_n |{\cal L}], N \mapsto [\sigma_1,..., \sigma_m| {\cal L}], N{++}[Split, rule(N), n]$, where $CondSplit$ (see below).

 \item [Fail] $[{\cal E}|{\cal L}], N \mapsto {\cal L}, N{++}[Fail, rule(N), n]$, where $CondFail$ (see below).
\end{description}
\end{minipage}
}}
\end{center}

The conditions appearing in our observation semantics are defined as follows: 

$SolveCond$:  $c$ is built-in, and $A' = wakeup(S, c,B)$ defines which CHR constraints of $S$ are woken by adding the constraint $c$ to the built-in store $B$.

$CondReac$:  the function $wake: Constraint, Trace \mapsto Wake$  is responsible for selecting the Wake event that justifies the Reactivate.

$CondApp1$: where the $j^{th}$ occurrence of a constraint with same functor as $c$  exists in the head of a fresh variant of a rule $r @ H'_1 \backslash H'_2 \Leftrightarrow g | C$ and a matching substitution $e$, such that $chr(H_1) = e(H_1')$, $chr(H_2) = e(H_2')$ and  $\{(r,id(H_1){++}id(H_2))\} \notin T$. 

$CondApp2$: $C$ is the body of the rule $r @ H'_1 \backslash H'_2 \Leftrightarrow g | C$.  The
$tryRule: Trace \mapsto TryRule$ will retrieve the TryRule event generated by Apply.1. It will search for the event in the trace log, normally the event TryRule will be one step back.  $addRDCs: Body \mapsto Sequence(RDC)$ will select only the RDCs on the body; the function $addBICs: Body \mapsto Sequence(BIC)$ will select the BICs on the body. Same conditions of $Apply.1$  plus  $CT \models \exists (B) \land \forall B \supset \exists (e \land g)$ , $T' = T \cup \{(r,id(H_1){++}id(H_2))\}$, and $H = c\#i:j$ if $c$ is in $H_1$, $H = []$ if $c$ is in $H_2$, 
and $c \in H_1 \underline{\vee} H_2.$

$CondSplit$:  where $\sigma_i = \tuple{[A_i|A], S, B, T}_{n} $, for $1 \leq i \leq m$, and $rule: Trace \mapsto ApplyRule$ is a function that will retrieve the cause of the split, a  disjunctive rule.

$CondFail$: $n$ is the numbering of the failed state ${\cal E}$, and $rule(N)$ is a function that will retrieve the cause of the failure.

\subsection{OS-\CHRv: Reconstruction ($I_l$)}
\label{rec-omega-t}

We show here, by specifying a local reconstruction function $I_l$, that the observationa semantics is faithful, i.e. that the extracted (actual) generic trace contains all information needed to reconstruct the original virtual trace (the trace semantics equivalente to the given operational semantics).

\vspace{1mm}
The local reconstruction function takes a current state $s$ and the generated actual trace including the last generated event $a$, $N {++} a$; it identifies the type of action $r$ and produces the new virtual state $s'$. 
It has the form, for each type of action $r$: $I_l(s, N a) = (r, s')$, which will be represented by a rule:
\begin{center}
{\bf r} $s, N {++} a \mapsto s'$
\end{center}
with $a = [r | a']$.

\vspace{1mm}
The initial configuration will be represented as ${\cal L}_0 = [\tuple{A, \{\}, true, \{\}}_{1}] ,N=[]$

\vspace{1mm}
For each item corresponding to a type of action, the reconstructed state $s'$ is the same as the new state obtained by the corresponding transition $Tr(r,s) = s'$ as described in the table~\ref{ssec:opsem}. This constitutes the proof of faithfulness of OS-\CHRv, by property~\ref{prop:faith}.


\vspace{1mm}
The conditions consists of an optional condition part and a computation part. They are just used here to express computations.

\vspace{1mm}
\begin{center}
{\small
\fbox{
\begin{minipage}[c]{12cm}
\begin{description}
 \item [Solve+Wake]$\\$
$[\tuple{A, S, B, T}_{n} | {\cal L}], N{++}[Wake, c, C, n'] \mapsto [\tuple{C+A',S, c\land B, T}_{n'} | {\cal L}]$, if $SolveCond$.

 \item [Activate] $[\tuple{A, S, B, T}_{n} | {\cal L}], N{++}[ActivateRDC, c, n'] \mapsto [\tuple{[c\#n':1|A'], c \uplus S,  B,T}_{n+1} | {\cal L}]$, if $ActiCond$.

 \item [Reactivate] $[\tuple{A, S, B, T}_n | {\cal L}], N{++}[ReactivateRDC, {c\#i}, wakeEvent, n'] \mapsto [\tuple{[c\#n':1|A'], \{c\#n'\} \uplus S,  B, T}_{n'} | {\cal L}]$, if $CondReac$.

 \item [Apply.1] $[\tuple{A, S, B, T}_n |{\cal L}], N{++}[TryRule, r, {c\#i:j}, H_1, H_2, g, n'] \mapsto [\tuple{[c\#i:j|A'], H_1 \uplus H_2 \uplus S', B, T}_n | {\cal L}]$, if $CondApp1$.

 \item [Apply.2]  $[\tuple{A, S, B, T}_n| {\cal L}], N{++}[ApplyRule, r, RDCs, BICs, H_1, H_2, e, H, n'] \mapsto [\tuple{C+A, H_1 \uplus S', e \wedge B, T'}_n| {\cal L}]$, if $CondApp2$.

 \item [Drop] $[\tuple{A, S, B, T}_n | {\cal L}], N {++} [Drop, {c\#i:j}, n'] \mapsto [\tuple{A', S, B, T}_{n'} | {\cal L}]$, if $DropCond$.

 \item [Default] $[\tuple{A, S, B, T}_n | {\cal L}], N{++}[Default, {c\#i:j}, j', n'] \mapsto [\tuple{[{c\#i:j'}|A'], S, B, T}_{n'} | {\cal L}]$, if $DefCond$.

 \item [Split] $[\tuple{A, S, B, T}_n | {\cal L}], N{++}[Split, r, n'] \mapsto [\sigma_1,..., \sigma_m|{\cal L}']$, if $CondSplit$.

 \item [Fail] $[\tuple{A, S, B, T}_n | {\cal L}], N{++}[Fail, r, n'] \mapsto {\cal L}$, if $CondFail$.
\end{description}
\end{minipage}
}}
\end{center}

$SolveCond$:  $A=[c|A'] \wedge n=n' \wedge C=wakeup(S,c,B)$.\\

$ActiCond$: $A=[c|A'] \wedge n=n'$. \\

$CondReac$: $A=[c\#i|A'] \wedge wakeEvent=wake(c,N) \wedge n=n'$.\\

$CondApp1$: $A=[c\#i:j|A'] \wedge S=H_1 \uplus H_2 \uplus S' \wedge n'=n$. \\

$CondApp2$: $A = H+A'$, $C$ is the body of the rule $\langle r @ H'_1 \backslash H'_2 \Leftrightarrow g | C \rangle \wedge S= H_1 \uplus H_2 \uplus S' \wedge T' = T \cup \{(r,id(H_1)+id(H_2))\} \wedge n'=n$, and $H = c\#i:j$ if $c$ is in $H_1$, $H = []$ if $c$ is in $H_2$ and $c \in H_1 \underline{\vee} H_2.$. \\

$DropCond$: $A = [c\#i:j|A'] \wedge n'=n$. \\

$DefCond$: $A= [c\#i:j|A'] \wedge j'=j+1 \wedge n'=n$. \\

$CondSplit$: $A = \tuple{[A_1\lor ... \lor A_m|A'], S, B, T}_n \wedge \sigma_i = \tuple{[A_i|A], S, B, T}_{n} $, for $1 \leq i \leq m \wedge n'=n$. \\

$CondFail$: $n'=n$.

%
%

\vspace{2mm}
Notice that the local reconstruction function does not use the full generated trace, but the last actual trace event only. This guarantees a better efficiency for analysis. However this is obtained at the cost of including specific details into the generic trace. For example the attribute $n$ (state indicator) could be retrieved from the current virtual state; the fact to have it in the current actual trace event avoids some computation.

\lstnewenvironment{prolog}[1][]
{\lstset{language=Prolog,frame=single,#1}}
{}

\section{Prototyping the Operational Semantics of \CHRv}
\label{sec:osproto}

The objective of building a prototype of the operational semantics described in the section~\ref{ssec:opsem} is to produce tests of trace generation, in order to improve the quality of the design. In fact there is no way to prove that the specification given in an algebraic style, with some implicit or informal parts, is sound. Indeed the simple fact to implement in Prolog an executable specification with which it is possible to simulate the extraction of the generic trace is a step towards a better quality of the proposed formal specification.

\vspace{1mm}
The SWI-Prolog implementation of the operational semantics of \CHRv\ is built with the feature
to produce the generic trace. The following subsections illustrate the architecture of the proposed implementation. The full program is given in the annexe.

\subsection{\CHRv\ Syntax and Rule Compilation}
\label{SecCHRvSyntax}

\begin{center}
\begin{minipage}{5cm}
\begin{prolog}[caption=\CHRv\ Syntax,label=CHRvSyntax]
:- op(1100,xfx,\ ).
:- op(1180,xfx,==>).
:- op(1180,xfx,<=>).
:- op(1190,xfy,@).
\end{prolog}
\end{minipage}
\end{center}
The \CHRv\ Syntax is represented Prolog infixed operators as presented in Listing \ref{CHRvSyntax}. This
syntax accepts \CHRv\ rules like in the following CHR program:

\begin{prolog}[frame=tb]
transitivity @ leq(vX,vY), leq(vY,vZ) ==> leq(vX,vZ).
idempotency @ leq(vX,vY) / leq(vX,vY) <==> true.
antisymmetry @ leq(vX,vY) , leq(vY,vX) <==> vX=vY.
\end{prolog}
However, the \CHRv\ Operational Semantics does not search rules individually. It searches for the i-th occurrence of some 
constraint operator in program. In order to provide the appropriate information, the \CHRv\ implementation 
generates a compiled version of this \CHRv\ program as presented below:

\begin{prolog}[frame=tb]
rule(leq,1,transitivity,[vX,vY],
     [],[leq(vX,vY),leq(vY,vZ)],[],[leq(vX,vZ)]).
rule(leq,2,transitivity,[vY,vZ],
     [],[leq(vX,vY),leq(vY,vZ))],[],[leq(vX,vZ)]).
rule(leq,3,idempotency,[vX,vY],
     [leq(vX,vY)],[leq(vX,vY)],[],[true]).
rule(leq,4,idempotency,[vX,vY],
     [leq(vX,vY)],[leq(vX,vY)],[],[true]).
rule(leq,5,antisymmetry,[vX,vY],
     [active,leq(vY,vX)],[],[],[vX=vY])
rule(leq,5,antisymmetry,
     [vY,vX],[leq(vX,vY),leq(vY,vX)],[],[],[vX=vY])
\end{prolog}
the predicate \verb+rule(Op,Ind,Rule,Args,Remove,Keep,Guard,Body)+ identifies the \CHRv\ rule that contains the \verb+Ind+th occurrence
of \verb+Op+ in its head.

\subsection{Auxiliary Functions}

To implement \CHRv\ Operational Semantics, it was necessary to define some
auxiliary functions responsible for variable management and built-ins implementation.

\subsubsection*{Variables}

In the proposed \CHRv\ implementation variables can be classified as local or global. Local
variables are used in \CHRv\ Rules and should be replaced by global variables or constants when the rule
is executed. \CHRv\ variables are not directly implemented using Prolog's variables because it was desirable to give more informations about variable names to trace and the search method.

In the implementation, variable can be any Prolog's atom. The predicates \verb+global_variable\1+ and \verb+local_variable+
are used to decide if some term represents a global or local variable, respectively. In the initial implementation local variables 
are atoms whose first letter is ``v'' and the second is a upper-case letter. Global variables are terms in the form \verb+v(_)+ or atoms 
initiated with ``v'', followed by a lower-case letter. 

Local variables should be replaced by constants or global variables when a \CHRv\ is executed. The predicate: \verb+instantiate_locals(T1,T2,Binds)+
replaces the local variables found in \verb+T1+ by other values, producing the term \verb+T2+; the argument \verb+Binds+ contains
the performed substitutions. In \verb+T1+ local variables are not replaced by a constant, the \verb+instantiable_locals+ predicate replaces
the local variable by an undefined prolog variable. The auxiliary predicate \verb+allocate_unusedvars+, instantiates these
undefined variables with free global variables.

\subsubsection*{Built-ins Constraint}

The equality constraint (\verb+=+) is the only built-in constraint implemented. The built-in memory contains a list of 
normalized equalities in the form:\\ 
\verb+GlobalVariable = (GlobalVariable | Constant)+.\\
To manipulate the built-in memory, two auxiliary functions are defined:

\begin{itemize}
\item \verb+solve_builtin(C,B1,B2,UDC1,UDC2)+: inserts the built-in constraint \verb+C+ in the built-in memory \verb+B1+,
producing the built-in memory \verb+B2+. If the memory becomes inconsistent, the resultant memory produced is \verb+false+. 
When the built-in constraint is inserted in memory, the User Defined Constraints Memory UDC1 is analyzed and the constraints
affected by the built-in insertion are listed in \verb+UDC2+.
\item \verb+replace(T1,B,T2)+: replaces the global variables in T1 by their values (if they are defined), 
producing the term T2.
\end{itemize}

\subsection{\CHRv\ Operational Semantics Implementation}

The \CHRv\ Operational Semantics implementation is similar to the semantics proposed in Section 3.2. To illustrate this implementation,
Listing~\ref{PrologSolveWake} presents the Prolog implementation of the rules {bf Solve+Wake} and {\bf Activate}. 

\begin{center}
\begin{minipage}{10cm}
\begin{prolog}[caption=Solver+Wake Rule, label=PrologSolveWake]
%Solve+Wake Rule
[([C|A],UDC,B,H,N)|Tl]-->[(A2,UDC2,B2,H,N)|Tl] :-
       is_builtin(C),
       solve_builtin(C,B,B2,UDC,Wokeup),
       merge(Wokeup,A,A2),
       remove(UDC,Wokeup,UDC2).

%Activate Rule
[([C|Tl],UDC,B,H,N)|T] ---> ([([(C3#0)|Tl],[C3|UDC],B,H,N2)]|T) :-
       is_UDC(C),
       replace(C,B,C2),
       C3 = C2^N,
       N2 is N + 1.
\end{prolog}
\end{minipage}
\end{center}

\subsection{Generic Trace Extraction}

To generate the generic trace in the executable operational semantics, the predicate: \verb+gentrace(Id,Data)+ in added. It generates a trace with the information contained in \verb+Data+. The \verb+Id+ represents the ordering number of
the generated trace. To illustrate this, Listing~\ref{DecoratedPrologSolveWake} contains the modified version of the rules presented in Listing~\ref{PrologSolveWake}.

\begin{center}
\begin{minipage}{10cm}
\begin{prolog}[caption=Solver+Wake Rule, label=DecoratedPrologSolveWake]
%Solve+Wake Rule
[([C|A],UDC,B,H,N)|Tl] ---> ([(A2,UDC2,B2,H,N)|Tl])) :-
       is_builtin(C),
       solve_builtin(C,B,B2,UDC,Wokeup),
       insert_reference(Wokeup,Pos,RefWokeup),
       merge(RefWokeup,A,A2),
       remove(UDC,Wokeup,UDC2),
       gentrace(Pos,[wake,C,[wokeup,Wokeup]]).

%Activate Rule
[([C|Tl],UDC,B,H,N)|T] ---> ([([(C3#0)|Tl],[C3|UDC],B,H,N2)|Tl]) :-
       isUDC(C),
       replace(C,B,C2),
       C3 = C2^N,
       N2 is N + 1,
       gentrace(_,[activate,C,N]).
\end{prolog}
\end{minipage}
\end{center}

\subsection{Executing a \CHRv\ Program}

The \CHRv\ Operational Semantics is a small-step semantics: it describes a single computation step in program execution.
To execute a \CHRv\ program, the implementation defines the operator (\verb+--*->+) that executes continuously the Operational Semantics Operator (\verb+--->+).

\section{Prototyping of the generic \CHRv\ Trace Using SWI Prolog}
\label{sec:proto}

A generic \CHRv tracer for SWI-Prolog was developed, as the default trace output contains most of the necessary info to build the GT. In Section~\ref{sec61}, we introduce the SWI Prolog debug output trace produced when executing CHR rule bases. In the section~\ref{sec62}, we explain the SWI CHR trace; Section~\ref{ssec:oschrproto} presents the way to map the SWI produced trace into OS-\CHRv, the actual generic trace.In the last two sections we present some trace queries and an example\footnote{Source-code available on \url{http://www.assembla.com/code/generic-tracer/subversion/nodes}}.

\subsection{Running Example}
\label{sec61}

The generic trace will be illustrated on a simple disjunctive graph-coloring problem. The following \CHRv rules define a graph coloring solution:

\begin{verbatim}
node1@ node(r1,C) ==> (C = r ; C = b ; C = g).
node2@ node(r2,C) ==> (C = b ; C = g).
node3@ node(r3,C) ==> (C = r ; C = b).
node4@ node(r4,C) ==> (C = r ; C = b).
node5@ node(r5,C) ==> (C = r ; C = g).
node6@ node(r6,C) ==> (C = r ; C = g; C = t).
node7@ node(r7,C) ==> (C = r ; C = b).
startGraph@ edges<=> edge(r1,r2), edge(r1,r3), edge(r1,r4),
edge(r1,r7), edge(r2,r6), edge(r3,r7), edge(r4,r5), edge(r4,r7),
edge(r5,r6), edge(r5,r7).
wrong@ edge(Ri,Rj), node(Ri,Ci), node(Rj,Cj) ==> Ci = Cj | false.
l1@ l([ ],[ ]) <=> true.
l2@ l([R|Rs],[C|Cs]) <=> node(R,C), l(Rs,Cs).
\end{verbatim}

This CHR base handles a graph-coloring problem with at most 3 colors where any two nodes connected by a common edge must not have the same color. The constrain \verb|node(r1,C)| means that node r1 has color C, the \verb1startGraph1 rule defines edges between nodes of a graph and the \verb1wrong1 rule assures that two nodes will have different colors. A small part of the SWI trace of the execution of the following goal ``edges, l([r1,r7,r4,r3,r2,r5,r6],[C1,C7,C4,C3,C2,C5,C6])." is depicted here:

\begin{verbatim}
CHR:   (1) Insert: node(r1,_G9234) # <384>
CHR:   (2) Call: node(r1,_G9234) # <384>
CHR:   (2) Try: node(r1,_G9234) # <384> ==> _G9234=r;_G9234=b;_G9234=g.
CHR:   (2) Apply: node(r1,_G9234) # <384> ==> _G9234=r;_G9234=b;_G9234=g.
...
CHR:   (2) Insert: node(r7,_G9235) # <386>
CHR:   (3) Call: node(r7,_G9235) # <386>
CHR:   (3) Try: node(r7,_G9235) # <386> ==> _G9235=r;_G9235=b.
CHR:   (3) Apply: node(r7,_G9235) # <386> ==> _G9235=r;_G9235=b.
CHR:   (4) Wake: node(r7,r) # <386>
CHR:   (4) Try: node(r1,r) # <384>, edge(r1,r7) # <376>,
                node(r7,r) # <386> ==> r=r | false.
CHR:   (4) Apply: node(r1,r) # <384>, edge(r1,r7) # <376>,
                  node(r7,r) # <386> ==> r=r | false.
CHR:   (3) Fail: node(r7,r) # <386>
CHR:   (4) Wake: node(r7,b) # <386>
...
\end{verbatim}

This subset of the execution is responsible for trying the value C1 and C7 as red then backtracking because C1 and C7 cannot have the same colors. 

Informal definitions of the trace events of SWI-Prolog can be found here\footnote{\url{http://www.swi-prolog.org/pldoc/doc_for?object=section(2,'7.4', swi('/doc/Manual/debugging.html'))}}.  Some problems occur when an analysis of the trace is needed: the try/apply transition has no rule name, it's very difficult to link the name of the generated var with the name of the variable passed as goal since all vars were renamed and there isn't an efficient way to query it.  

\subsection{Understanding SWI-Prolog Trace}
\label{sec62}

The SWI-Prolog debugging output will produce a trace according to the following grammar depicted in table~\ref{table:swiGrammar}.
\begin{table}[t] 
\begin{tabular}[t]{|l|l|l|}
\hline SwiTrace & $T$ & ::= $\{P_1...P_m\}, m \geq 0$\\
\hline Ports & $P$ & ::= $C | E | F | R| W | I | RE | TY| A$  \\
\hline Call & $C$ & ::= $"CHR: (depth)  Call:"  CT$  \\
\hline Exit & $E$ & ::= $"CHR: (depth)  Exit:"  CT$  \\
\hline Fail & $F$ & ::= $"CHR: (depth)  Fail:"  CT$  \\
\hline Redo & $R$  & ::= $"CHR: (depth)  Redo:"  CT$  \\
\hline Wake & $W$ & ::= $"CHR: (depth)  Wake:"  CT$  \\
\hline Insert & $I$ & ::= $"CHR: (depth)  Insert:"  CT$  \\
\hline Remove & $RE$ & ::= $"CHR: (depth)  Remove:"  CT$  \\
\hline Constraint & $CT$ & ::= $constraintName(t_1...t_n) "\# <id>" $  \\
\hline Try & $TY$ & ::= $TY_{propagation} | TY_{simplification} | TY_{simpagation}$\\
\hline Try2 & $TY_{propagation}$ & ::= $"CHR: (depth)  Try:" H_k "==>" G "|" B |$\\ 
\hline Try3 & $TY_{simplification}$ & ::= $"CHR: (depth)  Try:" H_r "<=>" G "|" B |$  \\
\hline Try4 & $TY_{simpagation}$ & ::= $"CHR: (depth)  Try:"  H_k\backslash H_r "<=>" G "|" B $\\
\hline Apply & $A$ & ::= $A_{propagation} | A_{simplification} | A_{simpagation}$ \\
\hline Apply2 & $A_{propagation}$ & ::= $"CHR: (depth)  Apply:" H_k "==>" G "|" B |$\\ 
\hline Apply3 & $A_{simplification}$ & ::= $"CHR: (depth)  Apply:" H_r '<=>' G "|" B |$  \\
\hline Apply4 & $A_{simpagation}$ & ::= $"CHR: (depth)  Apply:"  H_k '\backslash' H_r "<=>" G "|" B $\\
\hline Head & $H$ & ::= $CT$ | $CT "," H $  \\
\hline Constraint2 & $CT_2$ & ::= $constraintName(t_1...t_n)$  \\
\hline Body or BIC& $G,B$ & ::= $CT_2 $ | $CT_2 "," G$   \\
\hline
\end{tabular}
\caption[Swi-Trace's Grammar]{Swi-Trace's Grammar}
\label{table:swiGrammar}  
\end{table}

SWI-Prolog's default search strategy implemented is depth-first, the parameter $depth$ indicates the transaction's actual level in the search tree and $id$ is the constraint's unique identifier. Small parts of the trace will be shown and explained.
\begin{verbatim}
CHR: (0) Insert: edges # <372>
CHR: (1) Call: edges # <372>
\end{verbatim}
The trace produced by these two ports are responsible for removing a constraint from the goal and insert in execution stack. Notice that in SWI's trace they always appear together.
\begin{verbatim}
CHR: (2) Exit: edge(r1,r2) # <373>
\end{verbatim}
The computation over the active constraint is finished.
\begin{verbatim}
CHR: (3) Try: node(r7,_G9235) # <386> ==> _G9235=r;_G9235=b.
CHR: (3) Apply: node(r7,_G9235) # <386> ==> _G9235=r;_G9235=b.
\end{verbatim}
The trace produced by Try and Apply ports only happens together and it means that a rule was tried and applied respectively.
\begin{verbatim}
CHR: (4) Wake: node(r7,r) # <386>
\end{verbatim}
The Wake port is traced when a built-in is solved, in this case the constraint was reactivated because C7 = r.

\subsection{Transforming SWI Tracer into OS-\CHRv}
\label{ssec:oschrproto}

The SWI's output is not enough to perform a translation to OS-\CHRv . We do need information about what was the goal passed (to map all variables) and access to the source-code (get the rule names). The inputs and outputs of the algorithm is illustrated by figure~\ref{fig:traceTranslator}. The Translator's algorithm will be explained by example.

\begin{figure}[h]
\centering
\includegraphics[width=0.8\linewidth]{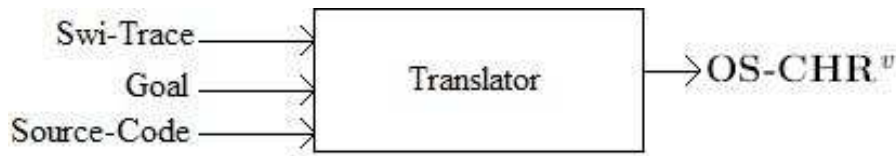}
\caption[Translator structure]
{Translator structure}
\label{fig:traceTranslator}  
\end{figure} 

For the Wake port the we have to look to previous values of the trace and determine what BIC solving fired this transition and also a Reactivate event will be produced. In this case \verb|C1 = a0| was the cause.
\begin{verbatim}
CHR:   (3) Wake: node(r1,a0) # <359>
-> [61,Wake,[=,C1,a0],[woken,[[node,r1,C1,359]]],360]
   ++
[62,Reactivate,[node,r1,a0,359],@61,360]
\end{verbatim}

Some ports have direct connection with OS-\CHRv: Call and Exit. All others ports will need a computation using the generated SWI trace. The Insert port is ignored because is redundant with the port Call.
\begin{verbatim}
CHR:   (1) Call: edges # <330>
-> 
GT: [0,ActivateRDC,[edges,330],331]

\end{verbatim}

For the tryRule map, we have to look the source code and try to find what is the rule name for that transition, and while generating the trace we keep track of the active constraint.
\begin{verbatim}
CHR:   (7) Try: node(r4,a0) # <371>, edge(r4,r5) # <340>, 
node(r5,a0) # <375> ==> a0=a0 | fail.
-> 
GT: [102,TryRule,failure@,[node,r5,a0,375],[keep,[[node,r4,a0,371], 
[edge,r4,r5,340], [node,r5,a0,375]]],[remove,[]],
[guard,[[=,a0,a0]]],376]
\end{verbatim}

ApplyRule is the most complicated map, we have to link(@) with the tryRule and check if it has a disjunctive body, if it is we have keep track to link correctly with a possible failure status; it can generate a lot of trace event depending on how many constraints were added/removed and possibly a split transition. The link function will recover the real name of the variable, in this case \_G9245 = C8
\begin{verbatim}
CHR:   (9) Apply: node(r8,_G9245) # <390> ==> _G9245=a0;
_G9245=a1;_G9245=a2;_G9245=a3.
-> 
[141,ApplyRule,@140,[addrdc],[addbic,[=,C8,a0],[=,C8,a1],
[=,C8,a2],[=,C8,a3]],[keep,[[node,r8,C8,390]]],
[remove,[]],[match,[node(_,C)=node(r8,C8)]],
[node,r8,C8,390],391]
++ 
GT: [142,Split, @141, 391]
\end{verbatim}

The Exit port has a direct map with OS-\CHRv.
\begin{verbatim}
CHR:   (2) Exit: edge(r1,r9) # <332>
->
GT: [4,Drop,[edge,r1,r9,332],333]
\end{verbatim}

The Fail port will produce a Fail event with its cause, a rule witch body contains the $false$ built-in.
\begin{verbatim}
CHR:   (6) Fail: node(r5,a0) # <375>
-> 
[109,Fail,@108,376]
\end{verbatim}

\subsection{Trace Querying}
\label{ssec:traceq}

The produced generic trace is represented by a sequence of java objects.
The language we choose for querying the trace is the SQL for Java Objects (JoSQL), its implementation can be found here\footnote{\url{http://josql.sourceforge.net/}}.

\vspace{1mm}
These are some examples of query in JoSQL: (on a trace of example~\ref{sec:proto})
\begin{itemize}
\item SELECT * FROM trace WHERE type ='{\tt ApplyRule}' AND (name='wrong@' OR name='node1@' OR name='node2@')
Will select the trace of the execution of rules: wrong, node1,node2.
\item SELECT * FROM trace WHERE type ='{\tt Split}' OR type ='{\tt Fail}'
Will select all split and fail transition.\\
\item SELECT addrdc,remove,addbic FROM trace WHERE type ='{\tt ApplyRule}'
\end{itemize}

The last query is more general and can by used by any application which need to handle a current state of the constraint store.
\subsection{Trace Analyzer}
\label{ssec:traceAnalyzer}
A Pretty Printer Analyzer was developed. This analyzer has a JoSQL query as parameter and prints the events that match the query.  The OS-\CHRv trace for leq with the goal  \textit{leq(A,B),leq(B,C),leq(C,A)}  is depicted. All events and attributes were selected by JoSQL, they are listed in Section \ref{gt-omega-t}. 

\begin{verbatim}
[0,ActivateRDC,[leq,A,B,201],202]
[1,Drop,[leq,A,B,201],202]
[2,ActivateRDC,[leq,B,C,202],203]
[3,TryRule,transitivity@,[leq,B,C,202],[keep,[[leq,A,B,201], 
 [leq,B,C,202]]], [remove,[]],[guard,[]],203]
[4,ApplyRule,@3,[addrdc,[leq,A,C]],[addbic],[keep,[[leq,A,B,201], 
[leq,B,C,202]]],[remove,[]],[match,[leq(X,Y)=leq(A,B)],
 [leq(Y,Z)=leq(B,C)]],[leq,B,C,202],203]
[5,ActivateRDC,[leq,A,C,204],205]
[6,Drop,[leq,A,C,204],205]
[7,Drop,[leq,B,C,202],205]
[8,ActivateRDC,[leq,C,A,205],206]
[9,TryRule,antisymmetry@,[leq,C,A,205],[keep,[]],
 [remove,[[leq,C,A,205], [leq,A,C,204]]],[guard,[]],206]
[10,ApplyRule,@9,[addrdc],[addbic,[=,C,A]],[keep,[]],
 [remove,[[leq,C,A,205], [leq,A,C,204]]],
 [match,[leq(X,Y)=leq(C,A)],[leq(Y,X)=leq(A,C)]],[leq,C,A,205],206]
[11,Wake,[=,C,A],[woken,[leq,B,C,202],[leq,A,B,201]],206]
[12,Reactivate,[leq,B,A,202],@11,206]
[13,TryRule,antisymmetry@,[leq,B,A,202],[keep,[]],
 [remove,[[leq,B,A,202], [leq,A,B,201]]],[guard,[]],206]
[14,ApplyRule,@13,[addrdc],[addbic,[=,B,A]],[keep,[]],
 [remove,[[leq,B,A,202],  [leq,A,B,201]]],
 [match,[leq(X,Y)=leq(B,A)],[leq(Y,X)=leq(A,B)]],[leq,B,A,202],206]
[15,Drop,[leq,A,A,202],206]
[16,Drop,[leq,A,A,205],206]
\end{verbatim}

%

\section{Experimentation}
\label{sec:exp}
To evaluate our approach 3 benchmarks were set: 10-Queens, primes\footnote{\url{http://people.cs.kuleuven.be/~tom.schrijvers/Research/CHR/chr_benchmarks/primes.chr}} and a compiled example of scheduling from CHORD\cite{msc-chord}, available on its test folder, the reason for choosing a CHORD example was the complexity, more than 100 rules. All results are shown in the following table.

All the experiments were performed on a PC with Pentium Core 2 Duo processor running at 2,4 GHz, with 4 GB of RAM and 1.5GB were reserved to the Java heap. The Prolog trace generator and the trace querying process are two different process as described by Langevine\cite{langevine08}.

\vspace{1mm}
The results are depicted in the table~\ref{lab-epx}. Each line corresponds to a program. The firts column gives the execution time without tracing (trace off); the second, the execution time with production of the SWI trace (trace on), the third the time in generic trace mode, and the last column gives the ratio between the sizes of both traces.

\vspace{3mm}
\begin{table} [h] 
\begin{tabular}{|l|l|l|l|l|}
  \hline
  \multicolumn{5}{|c|}{\CHRv tracing evaluation (time)} \\
  \hline
  \textbf{Problems} & No Trace & Swi Trace & \CHRt &\textbf{Size of the Trace}\\
  & & & & \textbf{(SWI/\CHRt)}\\ 
  scheduling & 0.1s  & 0.2s&0.27s & 0.5M / 0.5MB\\ 
   primes & 57s & 1min & 1min 05s & 5.4MB / 6.3MB\\
  10-Queens & 7s & 1 min 14s& 1min 25s& 59.7MB / 71.7MB \\
 graphColoring & 0.007s & 0.025s & 0.083s& 14.5KB / 21KB\\
  \hline
\end{tabular}\\\\
\caption{Experiments}
\label{lab-epx}
\end{table}

The following queries were done:
\begin{verbatim}
scheduling> g []. %starts chord computation
primes> candidates(8000). %calculates primes upto 8000
10-Queens> solveall(10,N,S). % give all solutions for 10 Queens.
graphColoring> edges, l([r1,r7,r4,r3,r2,r5,r6],[C1,C7,C4,C3,C2,C5,C6]). 
%graph with 7 edges
\end{verbatim}  

We observe that there is an (expected) slowdown in debugging modes (SWI trace and generic trace). It is slightly higher for the generic trace as it uses the SWI trace. Querying the generated trace does not slowdown more, since it can be done in paralel with trace generation.

For the selection of subtraces, tests executed on JoSQL show that that a list of 1,000,000 generic trace events can be queried in about 1.5s.

It must be noted that we limit the queries to patterns attached to a single event, limiting thus the complexity of the queries. Sophisticated queries envolving undetermined number of trace events could speed up seriously the performances.

\section{Discussion}
\label{sec:disc}

Several aspects of such a generic trace were explored on \cite{pierrerafRR09}, in particular its relations with components software development, the use of the fluent calculus to prototype traces and the use of object oriented specification methods. The generic trace presented in that work is thus limited to the simple theoretical operational semantics $\omega_t$ \cite{fruhwirth2003essentials} and therefore is less precise than the one given here.

Our approach of the observational semantics relies to abstract interpretation. The OS is similar to the ``Observable Semantics'' of Lucas \cite{lucas00} or the partial trace semantics of Cousot \cite{CousotPOPL02}. The parameters used to describe the execution states are, as expressed by Lucas, ``syntactic objects used to represent the conduct of operational mechanisms''. The traces are abstract representations of \CHRv\  semantics which allow to take into account the sole details we want to consider as common to different implementations. The (abstraction) relations between a generic trace and the traces of specific implementations of solvers are explored in \cite{cp11hal}, together with a compliance proof method.
Furthermore the generic trace contains a set of details considered as useful in several debugging tasks with several levels of refinement or observation. It could be enriched according to different needs~\footnote{An extensive study about the needs for constraint debugging can be found in \cite{DBLP:conf/discipl/2000}.}, or refined without changing the semantics of the already existing one. 

This way to proceed is opposite to the frequently adopted approach as, in particular, in \cite{simonisDFMQC10}, where a set of (visual) debugging tools is defined together with their input data, which consists of a restricted trace containing the minimal needed information. In our approach, we specify a semantically rich trace which can be used as input data for a potentially larger set of tools. 
The choice of the data to trace is made on the basis of a high level operational semantics, not on the basis of some specific debugging need. However the generic trace is designed in such a way that most of debugging tools devoted to the analysis of CHR resolution behavior may find in this trace what they need. As a consequence, based on this observational semantics, the work of implementation of the tracer and the work of designing debugging tools can be performed independently.

One may however feel that implementing a full generic trace is too much work demanding or that the resulting tracer performance will be considerably slow down. It has been shown in \cite{langevine08} that a generic approach may have more advantages than drawbacks in the sense that there may be a good trade-off between a very detailed generic trace (based on a more refined operational semantics) and the use of a trace driver able to query efficiently the generic trace, with a significant improvement in portability of debugging tools. We have shown here, that the implementation of the \CHRv\  generic trace in SWI-Prolog CHR implementation can easily be performed on the top of an existing tracer, resulting in a generic tracer practically as efficient as the original one on which it is based.

\section{Conclusion}
\label{sec:concl}
We have presented a first observational semantics of \CHRv (a formal specification of a \CHRv\  generic tracer), and two prototypes; the first is an executable operational semantics in Prolog which may produce a virtual generic trace; the second is a generic tracer of SWI CHR, based on the CHR SWI Prolog trace. The first helped to improve the quality of the formal operational semantics (in fact several corrections and/or improvements have been detected). The SWI CHR generic trace prototype shows that the generic trace can be easily and efficiently implemented on existing \CHRv\  implementations. The interest of the ``generic approach'' leads in the portability of analysis tools developed on the basis of this trace and the variety of possible trace based applications.

We do not claim that the \CHRv\  observational semantics which is presented here is the ultimate one. More refined observational semantics could be considered, including several levels of refinements (for example combining with Prolog semantics in Prolog based implementations); we just have shown that this approach can be realistic and useful in a great variety of CHR based software development.

Future work will concern more experimentation and improvements of the generic trace, OO based CHR implementation including a generic trace, and generic trace for hybrid constraint solvers.

\newpage

\begin{appendix}
\section{ANNEX: Operational Semantics in Prolog}
\label{ann:osproto}

The following is an implementation of the refined operational semantics of \CHRv\ in SWI-Prolog.  There are 6 files and 390 lines of Prolog's code.

\vspace{2mm}
\textbf{util.pl code:}
\begin{verbatim}
:- dynamic(val/2).

counter(Name,C) :- val(Name,C), !, retract(val(Name,C)), 
                   C2 is C + 1,
assert(val(Name,C2)).
counter(Name,0) :- assert(val(Name,1)).

% Auxiliary list functions

elem(E,[E|_]).
elem(E,[_|T]) :- elem(E,T).

subset([],_).
subset([H|T],L) :- elem(H,L),subset(T,L).

alldifferent([]).
alldifferent([H|T]) :- \+ elem(H,T), alldifferent(T).

remove([],_E,[]).
remove([X|T1],E,T2) :- (elem(X,E) -> T2=T3 ; T2=[X|T3]), 
                       remove(T1,E,T3).

insert_end(E,X) :- var(X), !, X=[E|_].
insert_end(E,X) :- X=[_H|T], insert_end(E,T).

closelist(X) :- var(X), !, X=[].
closelist(X) :- X=[_H|T], closelist(T).

at(M,X,Y) :- var(M), M=[(X=Y)|_].
at(M,X,Y) :- nonvar(M),M=[(X=Y2)|_], !, Y=Y2.
at(M,X,Y) :- nonvar(M),M=[(X2=_)|M2], X2\=X, at(M2,X,Y).

%% replace(T1,K,Val, T2) replaces all occuorence of K in T1 by Vals
%% resulting in term T2

replace(T1,Key,Val,T2) :- Key==T1, T2=Val.
replace(T1,Key,Val,T2) :- T1=..[F|Arg1],
replace_list(Arg1,Key,Val,Arg2), T2=..[F|Arg2].

replace_list([],_Key,_Val,[]).
replace_list([H1|T1],Key,Val,[H2|T2]) :-
replace(H1,Key,Val,H2),replace_list(T1,Key,Val,T2).

replace(T1,[],T1).
replace(T1,[(K=V)|M],T2) :- replace(T1,K,V,T_), replace(T_,M,T2).

%% For all local variables not initialized in the Head or Guard, 
%% it creates a new global variable to replace this local one

allocate_unusedvars([]).
allocate_unusedvars([(_K=V)|TL]) :- var(V), 
  counter(lastGlobalAllocated,N), V=v(N), allocate_unusedvars(TL).
allocate_unusedvars([(_K=V)|TL]) :- nonvar(V), 
                                    allocate_unusedvars(TL).


% Creates a Map that assigns all local_variables of T1 to
% an undefined value and T1 is the term T1 with these replacements
% applied

instantiate_locals(T1,T2,Map) :- local_variable(T1), !, 
                                 at(Map,T1,T2).
instantiate_locals(T1,T2,Map) :- T1=..[F|Arg1], 
instantiate_list_locals(Arg1,Arg2,Map), T2=..[F|Arg2].

instantiate_list_locals([],[],_).
instantiate_list_locals([H1|T1],[H2|T2],M) :-
      instantiate_locals(H1,H2,M),instantiate_list_locals(T1,T2,M).
\end{verbatim}

\textbf{Compiler code (compiler.pl):}

\begin{verbatim}
%% CHR Syntax
:- op(1100,xfx,\ ).
:- op(1180,xfx,==>).
:- op(1180,xfx,<=>).
:- op(1190,xfy,@).

% Transforms CHR sequences (A,B,C) in Prolog Lists [A,B,C]
makelist(true,[]) :- !.
makelist((X,Y),L) :- !,makelist(X,L1),makelist(Y,L2), merge(L1,L2,L).
makelist(L,[L]).

% testCase(ConstraintName,Index,Arguments,RuleName,Keep,Remove,
%          Guard,Body,returnFlag)
:- dynamic(rule/9).

% A Global variable a term v(N) or vU where U is not a uppercase 
% letter.
global_variable(v(_)) :- !.
global_variable(X) :- atom(X),atom_codes(X,[118,L|_]), (L<65 ; L>90).

% A Local variable a term vUxxx where U is an uppercase letter
local_variable(X)  :- atom(X),atom_codes(X,[118,L|_]), L>=65, L=<90.


calculateGuardBody(G|B,G,B) :- !.
calculateGuardBody(B,true,B).

%dbg(P) :-P, writef('DBG: %t\n',[P]).
%dbg(P) :-writef('start %t\n',[P]), P.
%dbg(P) :-writef('START: %t\n',[P]), P,writef('END: %t\n',[P]),!.
%dbg(P) :-writef('FAIL: %t\n',[P]),false.
dbg(P) :-P.

removeE([],_,[]).
removeE([E|T],E,T2) :- !,remove(T,E,T2).
removeE([X|T],E,[X|T2]) :- remove(T,E,T2).


% loop(Var,List,Pred)
loop(_,[],_).
loop(Var,[H|_],P) :- Var=H, P.
loop(Var,[_|T],P) :- loop(Var,T,P).

simpagation(N,Keep,Remove,Guard,Body):- 
  (N @ K \ R <=> G | B) ,  makelist(K,Keep),
  makelist(R,Remove), makelist(G,Guard), makelist(B,Body).
simpagation(N,Keep,Remove,[],Body)   :- (N @ K \ R <=> B) , 
 B\=(_|_), makelist(K,Keep),  makelist(R,Remove), makelist(B,Body).
simpagation(N,[],Remove,Guard,Body)  :- (N @ R <=> G | B) ,  
 R\=(_\_), makelist(R,Remove), makelist(G,Guard), makelist(B,Body).
simpagation(N,[],Remove,[],Body)     :- (N @ R <=> B) ,  
        R\=(_\_), B\=(_|_), makelist(R,Remove), makelist(B,Body).
simpagation(N,Keep,[],Guard,Body)    :- (N @ K ==> G | B),  
          makelist(K,Keep), makelist(G,Guard),  makelist(B,Body).
simpagation(N,Keep,[],[],Body)       :- (N @ K ==> B) , 
                   B\=(_|_),  makelist(K,Keep), makelist(B,Body).

compile :- retractall(rule(_,_,_,_,_,_,_,_,_)), 
     retractall(val(_,_)), simpagation(N,Keep,Remove,Guard,Body),
     ((elem(Active,Keep), InRemove=false);
        (elem(Active,Remove),InRemove=true)),
     Active=..[Op|Args], counter(Op,Index),
     assert(rule(Op,Index,Args,N,Keep,Remove,Guard,Body,InRemove)),
     fail.
compile.
\end{verbatim}

\textbf{Code generator (codegenerator.pl):}
\begin{verbatim}
optimize(if([],C),C2) :- !,optimize(C,C2).
optimize(if(true,C),C2) :- !,optimize(C,C2).
optimize(if([H|T],C),if(H,C2)) :- !,optimize(if(T,C),C2).
optimize(if(X,C),if(X,C2)) :- !,optimize(C,C2).

optimize(while([],C),C2) :- !,optimize(C,C2).
optimize(while(true,C),C2) :- !,optimize(C,C2).
optimize(while([H|T],C),while(H,C2)) :- !,optimize(while(T,C),C2).
optimize(while(X,C),if(X,C2)) :- !,optimize(C,C2).

optimize(seq([],B),B2) :- !, optimize(B,B2).
optimize(seq(B,[]),B2) :- !, optimize(B,B2).
optimize(seq(A,B),seq(A2,B2)) :- !, optimize(A,A2),optimize(B,B2).
optimize([H],H2) :- !, optimize(H,H2).
optimize([H|T],seq(H2,T2)) :- !, optimize(H,H2), optimize(T,T2).
optimize(X,X).

pprinter(if(T,C),Spaces) :- !,write(Spaces),write('if ('), 
   pprinter(T,''), write(') {\n'), concat(Spaces,'  ',Spaces2),
   pprinter(C,Spaces2), write('\n'),write(Spaces),write('}').
pprinter(while(T,C),Spaces) :- !,write(Spaces),write('while ('), 
   pprinter(T,''), write(') {\n'), concat(Spaces,'  ',Spaces2), 
   pprinter(C,Spaces2), write('\n'),write(Spaces),write('}').
pprinter(seq(X,Y),Spaces) :- !,pprinter(X,Spaces),nl,
   pprinter(Y,Spaces).
pprinter(comentario(C),Spaces) :- !, writef('%t//%t\n',[Spaces,C]).
pprinter(X,Spaces) :- write(Spaces),write(X).


printprogram :- 
 forall(constraint(Op), 
 (
  writef('void %t() {\n',[Op]),
  forall(testCase(Op,Args,N,B,C,D,E,Ret), 
   (
    (Ret=false -> Code = if(Args,while(B,while(C,if(D,E)))) ;
             Code = if(Args,while(B,while(C,if(D,seq(E,return)))))
	),
          optimize(seq(comentario(['from ',N]),Code),Code2),
	  pprinter(Code2,'   '),write('\n')
           )),
	   writef('\n}\n\n',[])
  )       ).
printprogram.

\end{verbatim}

\textbf{Interpreter (interpreter.pl):}
\begin{verbatim}

chr([],[]).
chr([(C^_)|T],[C|T2]) :- chr(T,T2).

id([],[]).
id([(_^I)|T],[I|T2]) :- id(T,T2).

chr_id(A,B,C) :- chr(A,B),id(A,C).

%
% Builtins Functions
%
%

% the only defined builtins is equality
is_builtin(_=_).


% uses prolog execution to evaluate guards.
check_guard([]).
check_guard([H|T]) :- H, check_guard(T).


dependon(X,X).
dependon(T,X) :- T=..[_|Args], some_depends(Args,X).

some_depends([],_X) :- fail.
some_depends([H|T],X) :- dependon(H,X) ; some_depends(T,X).

wakeup_policy([],_X,[]).
wakeup_policy([H1|T1],X,[H1|T2]) :- dependon(H1,X),!,
                                    wakeup_policy(T1,X,T2).
wakeup_policy([_H1|T1],X,T2) :- wakeup_policy(T1,X,T2).


% naive union_find algorithm to solve builtins.
%
solve_builtin(X=Y,B,B2,UDC,Wokeup) :- at(B,X,X2),!,
                              solve_builtin(X2=Y,B,B2,UDC,Wokeup).
solve_builtin(X=Y,B,B2,UDC,Wokeup) :- at(B,Y,Y2),!,
                              solve_builtin(X=Y2,B,B2,UDC,Wokeup).
solve_builtin(X=Y,B,B2,UDC,Wokeup) :- global_variable(X),!, 
                        B2=[(X=Y)|B], wakeup_policy(UDC,X,Wokeup).
solve_builtin(X=Y,B,B2,UDC,Wokeup) :- global_variable(Y),!, 
                        B2=[(Y=X)|B], wakeup_policy(UDC,Y,Wokeup).
solve_builtin(X=Y,B,B,_UDC,Wokeup) :- X==Y,!, Wokeup = [].
solve_builtin(X=Y,_B,fail,_UDC,[]) :- X\=Y.

%mywrite(X) :- write(traces,X).
mywrite(X) :- write(X).

start_tracing :- open('traces.txt',write,_,[alias(traces)]).
trace(TraceIdx, wake, Constraint, Woken) :- 
               gentrace(TraceIdx,[wake,Constraint,[woken,Woken]]).


gentrace(Idx,X) :- counter(trace_number,Idx),mywrite('['),
         mywrite(Idx),
	 forall(elem(N,X),(mywrite(','),mywrite(N))),write(']\n').

trace(X) :- mywrite(X), mywrite('\n').

stop_tracing :- close(traces).

:- op(100,xfx,#).

:- op(1200,xfx,--->).
:- op(1200,xfx,-*->).

% Operator that calculate one CHR Action
% Syntax:
% State --> [State]
% State = (Goal,UDConstraints,Built_ins,History,Index)
% Based on paper.

%(([C|Tl],UDC,B,H,N) --->(reactivate(C), [(Tl,[C2|UDC],B,H,N)])) :-
%       \+ is_builtin(C),
%       C = _^_,
%       replace(C,B,C2).
isUDC(X) :- (is_builtin(X); C=reactive(_,_); C=_#_),!,false.
isUDC(_).

insert_reference([],_R,[]).
insert_reference([H|T1],R,[reactive(H,R)|T2]) :- 
                                        insert_reference(T1,R,T2).

(([(T1;T2)|A],UDC,B,H,N) ---> 
                   [([T1|A],UDC,B,H,N), ([T2|A],UDC,B,H,N) ]) :-
                                     gentrace(_,[split,T1,T2]).

((_Goal,_UDC,fail,_H,_N) ---> []) :- gentrace(_,[reject]).

(([C|A],UDC,B,H,N) ---> ([(A2,UDC2,B2,H,N)])) :-
       is_builtin(C),
       solve_builtin(C,B,B2,UDC,Wokeup),
       insert_reference(Wokeup,Pos,RefWokeup),
       merge(RefWokeup,A,A2),
       remove(UDC,Wokeup,UDC2),
       gentrace(Pos,[wake,C,[wokeup,Wokeup]]).

(([C|Tl],UDC,B,H,N) ---> ([([(C3#0)|Tl],[C3|UDC],B,H,N2)])) :-
       \+ is_builtin(C), \+ C=reactive(_,_), \+ C=_#_,
       replace(C,B,C2),
       C3 = C2^N,
       N2 is N + 1,
       gentrace(_,[activate,C,N]).

(([reactive(C,Ref)|Tl],UDC,B,H,N) ---> 
                              ([([(C2#0)|Tl],[C2|UDC],B,H,N)])) :-
	\+ is_builtin(C),
	C = _^_,
	replace(C,B,C2),
	gentrace(_,[reactivate,C,Ref]).

(([((C^Id)#Index)|Goal],UDC,B,History,N) ---> 
                                       ([(Goal2,UDC2,B,H2,N)])) :-
	B \= fail,
	C=..[Op|Args1],
	rule(Op,Index,Args2, RName1, K1, R1 ,G1,B1,ActiveInRemove),
	instantiate_locals([Args2,K1,R1,G1,B1],
                      [Args1,K2,R2,G2,B2],Binds),closelist(Binds),
	chr_id(K3,K2,IdKeep),chr_id(R3,R2,IdRem),
	subset(K3,UDC),subset(R3,UDC),	
	merge(IdKeep,IdRem,IdHead), alldifferent(IdHead),
        (IdRem=[] -> \+ elem((RName1,IdHead),History) ; true),
	gentrace(IdTry,
         [tryRule,RName1,C,Id,Index,Binds, (K2 / R2 <=> G2 | B2)]),
	check_guard(G2),
	allocate_unusedvars(Binds),
	(IdRem=[] -> (merge([(RName1,IdHead)],History,H2)) ; 
                                                      H2=History),
	remove(UDC,R3,UDC2), 
	(ActiveInRemove -> merge(B2,Goal, Goal2) ; 
                           merge(B2,[(C^Id)#Index|Goal],Goal2)),!,
	gentrace(_,[applyRule,RName1,IdTry]).

(([((C^Id)#Index)|Goal],UDC,B,H,N) ---> ([(Goal,UDC,B,H,N)])) :-
	C=..[Op|_], \+ rule(Op,Index,_,_,_,_,_,_,_),
	gentrace(_,[drop,(C^Id)#Index]),!.


(([(C^Id)#Index|Goal],UDC,B,H,N) ---> 
                        ([([((C^Id)#Index2)|Goal],UDC,B,H,N)])) :-
	gentrace(_,[default,(C^Id)#Index,Index2]), 
        Index2 is Index + 1.

% Operator -*->
% Executes chains of ---> operator until no more execution is 
% possible

([SH1|ST1] -*-> (S2)) :-
       dbg((SH1 ---> (LS1))),!,
       merge(LS1,ST1,STemp),
       (STemp -*-> (S2)).

([SH1|ST1] -*-> ([SH1|ST2])) :-
	%trace(reject),
	(ST1 -*-> (ST2)).

([] -*-> []).

prettyprinter_actions([]).
prettyprinter_actions([H|T]) :-
prettyprinter_action(H),prettyprinter_actions(T).
prettyprinter_action(apply2(N,B)) :-
       (N @ C), instantiate_locals(C,C2,B),
       writef("apply2 Rule-%t:\n   %t\n   %t\n",[N,C,C2]).
prettyprinter_action(activate(AC)) :- 
                                   writef("activate : %t\n",[AC]).
prettyprinter_action(reactivate(AC^_)) :- 
                                 writef("reactivate : %t\n",[AC]).
prettyprinter_action(solve(B,_)) :- writef("solve : %t\n",[B]).
prettyprinter_action(split(X,Y)) :- 
                                writef("split : (%t,%t)\n",[X,Y]).
prettyprinter_action(fail) :- writef("fail\n").
prettyprinter_action(tryanother) :- write('found solution\n').

printstates([]).
printstates([(_,UDCS,BUILTS,_,_)|T]) :-
       chr(UDCS,Constr),
       writef("Solution:\n     UDCS = %t\n     Buitins= %t\n",
[Constr,BUILTS]), printstates(T).


%
% The predicate run(Goal) executes the Goal and
% printes the executed Actions and Memory.
%

run(Goal) :- ([(Goal,[],[],[],0)] -*-> (Sts)),
             printstates(Sts).

\end{verbatim}
\textbf{Main (chr.pl):}

\begin{verbatim}
:- [util].
:- [compiler].
:- [codegenerator].
:- [interpreter].
\end{verbatim}

\textbf{Test code (test.pl):}
\begin{verbatim}
:- [chr].
r0 @ false <=> 1=0.
reflexivity @ leq(vX,vX) <=> true.
antisimetry @ leq(vX,vY), leq(vY,vX) <=> vX=vY.
idempotency @ leq(vX,vY) \ leq(vX,vY) <=> true. 
transitivity @ leq(vX,vY), leq(vY,vZ) ==> leq(vX,vZ).
exemploLeq1 :- run([leq(v1,v2),leq(v2,v3),leq(v3,v1)]).
r3 @ candidate(vN) <=> vN>1, vM is vN - 1 |  
                                          prime(vN),candidate(vM).
r4 @ candidate(1) <=> true.
r5 @ prime(vX) \ prime(vY) <=> 0 is mod(vY,vX) | true.
exemploPrime1 :- run([candidate(50)]).
r6 @ color(vX) <=> vX=r ; vX=g ; vX=b .
r7 @ edge(vX,vY) ==> color(vY).
r8 @ edge(vX,vC1), edge(vY,vC2), link(vX,vY) ==> vC1=vC2 | false .
r9 @ graph <=> edge(1,vX), edge(2,vY), edge(3,vZ),link(1,2),
                                              link(1,3),link(2,3).
exemploGraph1 :- run([graph]).  
:- compile.
\end{verbatim}
\end{appendix}

\clearpage
\addcontentsline{toc}{section}{R\'ef\'erences}

\bibliographystyle{plain}
\bibliography{LTS,report}

\begin{thebibliography}{10}

\bibitem{gentra4cpa}
A.~{Aggoun}, T.~{Baudel}, P.~{Deransart}, M.~{Ducass\'e}, J.-D. {Fekete}, and
  N.~{Jussien}.
\newblock {Generic Trace Format for Constraint Programming, GenTra4CP, Version
  2.1}.
\newblock Technical report, {INRIA Paris-Rocquencourt}, {\'Ecole des Mines de
  Nantes}, {INSA de Rennes}, {Universit\'e d'Orl\'eans}, Cosytec and ILOG, July
  2004.
\newblock {\tt http://contraintes.inria.fr/OADymPPaC/Public/Trace/index.html}.

\bibitem{eclipse-prolog}
A.~{Aggoun {\& al.}}
\newblock {ECLiPSe User Manual. Release 5.3}.
\newblock 2001.

\bibitem{CousotPOPL02}
P.~Cousot and R.~Cousot.
\newblock Systematic design of program transformation frameworks by abstract
  interpretation.
\newblock In {\em Proc. of POPL 2002}, pages 178--190, 2002.

\bibitem{msc-chord}
Marcos Aur\'elio~Almeida da~Silva.
\newblock {CHORD: Constraint Handling Object-oriented Rules with Disjunctions}.
\newblock Master's thesis, Universidade Federal de Pernambuco, February 2009.
\newblock Universidade Federal de Pernambuco.

\bibitem{kosd06}
L.~De~Koninck, T.~Schrijvers, and B.~Demoen.
\newblock {Search Strategies in CHR(Prolog)}.
\newblock 2006.
\newblock
  {http://citeseerx.ist.psu.edu/}{viewdoc/download?doi=10.1.1.69.7359\&rep=rep1\&type=pdf}.

\bibitem{cp11hal}
P.~{Deransart}.
\newblock {Generic Traces and Constraints, GenTra4CP Revisited}, May 2011.
\newblock {\tt http://hal.inria.fr/hal-00597033}.

\bibitem{TMTmanuscript11e}
P.~{Deransart}.
\newblock {Towards a Trace Meta-Theory}, July 2011.
\newblock Working document {\tt http://hal.inria.fr/inria-00443648} (almost in
  French).

\bibitem{DBLP:conf/discipl/2000}
Pierre Deransart, Manuel~V. Hermenegildo, and Jan Maluszynski, editors.
\newblock {\em Analysis and Visualization Tools for Constraint Programming,
  Constraint Debugging (DiSCiPl project)}, volume 1870 of {\em Lecture Notes in
  Computer Science}. Springer, 2000.

\bibitem{frue09}
Thom Fruehwirth.
\newblock {\em {Constraint Handling Rules}}.
\newblock Cambridge, 2009.

\bibitem{fruhwirth2003essentials}
T.~Fruhwirth and S.~Abdennadher.
\newblock {\em {Essentials of constraint programming}}.
\newblock Springer-Verlag New York Inc, 2003.

\bibitem{fru_brisset_highlevel_implementation_techrep95}
Thom Fr{\"u}hwirth and Pascal Brisset.
\newblock High-level implementations of {{C}onstraint {H}andling {R}ules}.
\newblock Technical Report ECRC-95-20, European Computer-Industry Research
  Centre, Munchen, Germany, 1995.

\bibitem{holzfrueSics98}
C.~Holzbaur and Thom Fr{\"u}hwirth.
\newblock {Constraint Handling Rules Reference Manual for Sicstus Prolog}, July
  1998.
\newblock Technical Report TR-98-01.

\bibitem{jahier00}
E.~Jahier, M.~Ducass\'e, and O.~Ridoux.
\newblock Specifying {Prolog} trace models with a continuation semantics.
\newblock In K.-K. Lau, editor, {\em Proc. of {LO}gic-based Program Synthesis
  and TRansformation}, London, July 2000.
\newblock Technical Report Report Series, Department of Computer Science,
  University of Manchester, ISSN 1361-6161. Report number UMCS-00-6-1.

\bibitem{ercimlnai}
Ludovic Langevine, Pierre Deransart, and Mireille Ducass\'e.
\newblock {A Generic Trace Schema for the Portability of CP(FD) Debugging
  Tools}.
\newblock In K.R. Apt, F.~Fages, F.~Rossi, P.~Szeredi, and Jozsef Vancza,
  editors, {\em Recent Advances in Constraints}, number 3010 in LNAI. Springer
  Verlag, May 2004.

\bibitem{langevine08}
Ludovic Langevine and Mireille Ducass\'e.
\newblock Design and implementation of a tracer driver: Easy and efficient
  dynamic analyses of constraint logic programs.
\newblock {\em Theory and Practice of Logic Programming, Cambridge University
  Press}, 8(5-6), Sep-Nov 2008.
\newblock {\tt http://arxiv.org/abs/0804.4116}.

\bibitem{lucas00}
Salvador Lucas.
\newblock {Observable Semantics and Dynamic Analysis of Computational
  Processes}.
\newblock Technical Report LIX/RR/00/02, Laboratoire d'Informatique LIX, 2000.

\bibitem{pierrerafRR09}
Rafael Oliveira and Pierre Deransart.
\newblock {Towards a Generic Framework to Generate Explanatory Traces of
  Constraint Solving and Rule-Based Reasoning}.
\newblock Technical Report RR-7165, INRIA Paris-Rocquencourt, December 2009.

\bibitem{simonisDFMQC10}
Helmut Simonis, Paul Davern, Jacob Feldman, Deepak Mehta, Luis Quesada, and
  Mats Carlsson.
\newblock {A Generic Visualization Platform for CP}.
\newblock In Karen Petrie, editor, {\em Proceedings of the 16th International
  Conference on Principles and Practice of Constraint Programming}, St Andrews,
  Scotland, September 2010.

\bibitem{wielemaker2006swi}
J.~Wielemaker.
\newblock {SWI-Prolog 5.6 Reference Manual}.
\newblock {\em Department of Social Science Informatics, University of
  Amsterdam, Amsterdam, Marz}, 2006.

\end{thebibliography}

\newpage
\tableofcontents

\end{document}